\documentclass{article}
\usepackage{amssymb,nicefrac}
\usepackage{amsfonts}
\usepackage{graphicx}
\usepackage[fleqn]{amsmath}
\usepackage[varg]{txfonts}
\usepackage{stmaryrd}

\usepackage{framed}
\usepackage{color}
\usepackage{ulem}
\usepackage{tikzfig}

\tikzstyle{env}=[copoint,regular polygon rotate=0,minimum width=0.2cm, fill=black]

\tikzstyle{every picture}=[baseline=-0.25em]
\tikzstyle{dotpic}=[scale=0.5]
\tikzstyle{diredges}=[every to/.style={diredge}]
\tikzstyle{dot graph}=[shorten <=-0.1mm,shorten >=-0.1mm,scale=0.6]
\tikzstyle{plot point}=[circle,fill=black,minimum width=2mm,inner sep=0]

\tikzstyle{braceedge}=[decorate,decoration={brace,amplitude=2mm,raise=-1mm}]
\tikzstyle{small braceedge}=[decorate,decoration={brace,amplitude=1mm,raise=-1mm}]
\tikzstyle{left hook arrow}=[left hook-latex]
\tikzstyle{right hook arrow}=[right hook-latex]

\tikzstyle{dtriangle}=[fill=yellow,draw=black,shape=isosceles triangle,shape border rotate=-90,isosceles triangle stretches=true,inner sep=1pt,minimum width=0.4cm,minimum height=3mm]
\tikzstyle{vtriang}=[fill=yellow,draw=black,shape=isosceles triangle,shape border rotate=180,isosceles triangle stretches=true,inner sep=1pt,minimum width=0.4cm,minimum height=3mm]
\tikzstyle{triangle}=[fill=yellow,draw=black,shape=isosceles triangle,shape border rotate=90,isosceles triangle stretches=true,inner sep=1pt,minimum width=0.4cm,minimum height=3mm]
\tikzstyle{H box}=[rectangle,fill=yellow,draw=black,xscale=1.2,yscale=1.2, inner sep=1.pt]

\tikzstyle{bn}=[circle,fill=black,draw=black,scale=.4]
\tikzstyle{wn}=[circle,fill=white,draw=black,scale=.6]
\tikzstyle{dn}=[circle,fill=none,draw=gray]

\tikzstyle{black dot}=[inner sep=0.7mm,minimum width=0pt,minimum height=0pt,fill=black,draw=black,shape=circle]

\tikzstyle{dot}=[black dot]
\tikzstyle{smalldot}=[inner sep=0.4mm,minimum width=0pt,minimum height=0pt,fill=black,draw=black,shape=circle]
\tikzstyle{white dot}=[dot,fill=white]
\tikzstyle{antipode}=[white dot,inner sep=0.3mm,font=\footnotesize]
\tikzstyle{smallwhitedot}=[smalldot,fill=white]
\tikzstyle{alt white dot}=[white dot,label={[xshift=3.07mm,yshift=-0.05mm,font=\footnotesize]left:$*$}]
\tikzstyle{gray dot}=[dot,fill=gray!40!white]
\tikzstyle{smallgraydot}=[smalldot,fill=gray!40!white]
\tikzstyle{box vertex}=[draw=black,rectangle]
\tikzstyle{small box}=[box vertex,fill=white]
\tikzstyle{whitebg}=[fill=white,inner sep=2pt]
\tikzstyle{graph state vertex}=[sg vertex,fill=black]

\tikzstyle{wide copoint}=[fill=white,draw=black,shape=isosceles triangle,shape border rotate=90,isosceles triangle stretches=true,inner sep=1pt,minimum width=1.5cm,minimum height=5mm]
\tikzstyle{wide point}=[fill=white,draw=black,shape=isosceles triangle,shape border rotate=-90,isosceles triangle stretches=true,inner sep=1pt,minimum width=1.5cm,minimum height=4mm]
\tikzstyle{very wide copoint}=[fill=white,draw=black,shape=isosceles triangle,shape border rotate=-90,isosceles triangle stretches=true,inner sep=1pt,minimum width=2.5cm,minimum height=4mm]
\tikzstyle{very wide empty copoint}=[draw=black,shape=isosceles triangle,shape border rotate=-90,isosceles triangle stretches=true,inner sep=1pt,minimum width=2.5cm,minimum height=4mm]
\tikzstyle{symm}=[ultra thick,shorten <=-1mm,shorten >=-1mm]

\tikzstyle{square box}=[rectangle,fill=white,draw=black,minimum height=5mm,minimum width=5mm,font=\small]
\tikzstyle{square gray box}=[rectangle,fill=gray!30,draw=black,minimum height=6mm,minimum width=6mm]
\tikzstyle{copoint}=[regular polygon,regular polygon sides=3,draw=black,scale=0.75,inner sep=-0.5pt,minimum width=7mm,fill=white]
\tikzstyle{point}=[regular polygon,regular polygon sides=3,draw=black,scale=0.75,inner sep=-0.5pt,minimum width=7mm,fill=white,regular polygon rotate=180]
\tikzstyle{gray point}=[point,fill=gray!40!white]
\tikzstyle{gray copoint}=[copoint,fill=gray!40!white]

\newcommand{\edgearrow}{{\arrow[black]{>}}}
\newcommand{\edgetick}{{\arrow[black,scale=0.7,very thick]{|}}}

\tikzstyle{diredge}=[->]
\tikzstyle{rdiredge}=[<-]
\tikzstyle{medium diredge}=[->]

\tikzstyle{short diredge}=[->]
\tikzstyle{halfedge}=[-)]
\tikzstyle{other halfedge}=[(-]
\tikzstyle{freeedge}=[(-)]
\tikzstyle{white edge}=[line width=5pt,white]
\tikzstyle{tick}=[postaction=decorate,decoration={markings, mark=at position 0.5 with \edgetick}]
\tikzstyle{small map edge}=[|-latex, gray!60!blue, shorten <=0.9mm, shorten >=0.5mm]
\tikzstyle{thick dashed edge}=[very thick,dashed,gray!40]
\tikzstyle{map edge}=[|-latex,very thick, gray!40, shorten <=1mm, shorten >=0.5mm]
\tikzstyle{tickedge}=[postaction=decorate,
  decoration={markings, mark=at position 0.5 with \edgetick}]
\tikzstyle{dirtickedge}=[postaction=decorate,
  decoration={markings, mark=at position 0.5 with \edgetick},
  decoration={markings, mark=at position 0.85 with \edgearrow}]
\tikzstyle{dirdoubletickedge}=[postaction=decorate,
  decoration={markings, mark=at position 0.4 with \edgetick},
  decoration={markings, mark=at position 0.6 with \edgetick},
  decoration={markings, mark=at position 0.85 with \edgearrow}]

\makeatletter
\newcommand{\boxshape}[3]{%
\pgfdeclareshape{#1}{
\inheritsavedanchors[from=rectangle] 
\inheritanchorborder[from=rectangle]
\inheritanchor[from=rectangle]{center}
\inheritanchor[from=rectangle]{north}
\inheritanchor[from=rectangle]{south}
\inheritanchor[from=rectangle]{west}
\inheritanchor[from=rectangle]{east}
\backgroundpath{
\southwest \pgf@xa=\pgf@x \pgf@ya=\pgf@y
\northeast \pgf@xb=\pgf@x \pgf@yb=\pgf@y

\@tempdima=#2
\@tempdimb=#3

\pgfpathmoveto{\pgfpoint{\pgf@xa - 5pt + \@tempdima}{\pgf@ya}}
\pgfpathlineto{\pgfpoint{\pgf@xa - 5pt - \@tempdima}{\pgf@yb}}
\pgfpathlineto{\pgfpoint{\pgf@xb + 5pt + \@tempdimb}{\pgf@yb}}
\pgfpathlineto{\pgfpoint{\pgf@xb + 5pt - \@tempdimb}{\pgf@ya}}
\pgfpathlineto{\pgfpoint{\pgf@xa - 5pt + \@tempdima}{\pgf@ya}}
\pgfpathclose
}
}}

\boxshape{NEbox}{0pt}{8pt}
\boxshape{SEbox}{0pt}{-8pt}
\boxshape{NWbox}{8pt}{0pt}
\boxshape{SWbox}{-8pt}{0pt}
\makeatother

\tikzstyle{map}=[draw,shape=NEbox,inner sep=7pt]
\tikzstyle{mapdag}=[draw,shape=SEbox,inner sep=7pt]
\tikzstyle{maptrans}=[draw,shape=SWbox,inner sep=7pt]
\tikzstyle{mapconj}=[draw,shape=NWbox,inner sep=7pt]

\tikzstyle{probs}=[shape=semicircle,fill=gray!40!white,draw=black,shape border rotate=180,minimum width=1.2cm]

\tikzstyle{arrs}=[-latex,font=\small,auto]
\tikzstyle{arrow plain}=[arrs]
\tikzstyle{arrow dashed}=[dashed,arrs]
\tikzstyle{arrow bold}=[very thick,arrs]
\tikzstyle{arrow hide}=[draw=white!0,-]
\tikzstyle{arrow reverse}=[latex-]
\tikzstyle{cdnode}=[]

\tikzstyle{gn}=[dot,fill=green,minimum width=0.3cm,inner sep=0pt]
\tikzstyle{rn}=[dot,fill=red,inner sep=0pt,minimum width=0.3cm]

\tikzstyle{rc}=[dot,thick,fill=white,draw = red,minimum width=0.3cm,inner sep=0pt]
\tikzstyle{gc}=[dot,thick,fill=white,draw= green,inner sep=0pt,minimum width=0.3cm]
\tikzstyle{bc}=[dot,thick,fill=white,draw= blue,minimum width=0.3cm]

\tikzstyle{label}=[circle,fill=white,minimum width=0.3cm]

\tikzstyle{clocklabel}=[dot,fill=yellow,draw=black,font=\tiny,inner sep=0.75pt]

\tikzstyle{rsn}=[circle split,draw,fill=red,font=\tiny,inner sep=0.75pt]
\tikzstyle{gsn}=[circle split,draw,fill=green,font=\tiny,inner sep=0.75pt]
\tikzstyle{bsn}=[circle split,draw,fill=blue,font=\tiny,inner sep=0.75pt]

\tikzstyle{rsc}=[circle split,thick,draw= red,draw,fill=white,font=\tiny,inner sep=0.75pt]
\tikzstyle{gsc}=[circle split,thick,draw= green,draw,fill=white,font=\tiny,inner sep=0.75pt]
\tikzstyle{bsc}=[circle split,thick,draw= blue,draw,fill=white,font=\tiny,inner sep=0.75pt]

\tikzstyle{cnot}=[fill=white,shape=circle,inner sep=-1.4pt]
\tikzstyle{wire label}=[font=\tiny, auto]

\newcommand{\bra}[1]{\ensuremath{\left\langle #1 \right|}}
\newcommand{\ket}[1]{\ensuremath{\left|  #1 \right\rangle}}

\tikzstyle{cdiag}=[matrix of math nodes, row sep=3em, column sep=3em, text height=1.5ex, text depth=0.25ex,inner sep=0.5em]
\tikzstyle{arrow above}=[transform canvas={yshift=0.5ex}]
\tikzstyle{arrow below}=[transform canvas={yshift=-0.5ex}]

\usepackage{mathtools}
\usepackage{placeins}
\newtheorem{Th}{Theorem}[section]
\newtheorem{theorem}[Th]{Theorem}
\newtheorem{proposition}[Th]{Proposition}
\newtheorem{lemma}[Th]{Lemma}

\newenvironment{proof}{\textbf{Proof:}}{\hfill$\Box$\newline}

\title{A universal completion of the ZX-calculus}
\author{Kang Feng Ng \qquad\qquad Quanlong Wang\\ Department of Computer Science, University of Oxford }

\begin{document}

\date{}\maketitle

\begin{abstract}
In this paper, we give a universal completion of the ZX-calculus for the whole of pure qubit quantum mechanics. This proof is based on the completeness of another graphical language: the ZW-calculus, with direct translations between these two graphical systems. 

\end{abstract}

\section{Introduction}
The ZX-calculus introduced by Coecke and Duncan \cite{CoeckeDuncan}  is an intuitive yet mathematically strict graphical language for quantum computing: it is formulated within the framework of compacted closed categories which has a rigorous underpinning for graphical calculus \cite{Joyal}, meanwhile being an important branch of categorical quantum mechanics (CQM)  pioneered by Abramsky and Coecke \cite{Coeckesamson}. Notably, it has intuitional and simple rewriting rules to transform diagrams from one to another. Each diagram in the ZX-calculus has a standard interpretation in the Hilbert spaces, thus makes  it relevant for quantum computing. For the past ten years, the ZX-calculus has enjoyed sucess in applying to fields of quantum information and quantum computation (QIC), in particular (topological) measurement-based quantum computing  \cite{Duncanpx,  Horsman}  and quantum error correction \cite{DuncanLucas,  ckzh}. 

To realise its greatest advantage, the so-called completeness is of concerned with the ZX-calculus:  any equation of diagrams that holds true under the standard interpretation in Hilbert spaces can be derived diagrammatically.  It has been shown in \cite{Zamdzhiev} that the original version of the ZX-calculus  \cite{CoeckeDuncan}  plus the Euler decomposition of Hadamard gate is incomplete for the  overall pure qubit quantum mechanics (QM). Since then, plenty of efforts have been devoted to completion of some part of QM:  real QM \cite{duncanperdrix}, stabilizer QM \cite{Miriam1}, single qubit Clifford+T QM \cite{Miriam1ct}  and Clifford+T QM \cite{Emmanuel}.  Amongst them, the completeness of ZX-calculus for Clifford+T QM is especially interesting, since it is approximatively universal for QM.  Note that their proof relies on the completeness of ZW-calculus for  "qubits with integer coefficients".

In this paper, we prove that the ZX-calculus is complete for the overall pure qubit QM.  Our proof is based on the completeness of  ZW-calculus for the whole qubit QM \cite{amar}:  we first introduce a triangle and a series of  $\lambda$-labeled boxes ($\lambda \geq 0$),  which turns out to be expressible in ZX-calculus  without these symbols.  Then we establish reversible translations from ZX to ZW and vice versa. By checking carefully that all the ZW rewriting rules still hold under translation from ZW to ZX, we finally finished the  proof of completeness of ZX-calculus for the overall qubit QM.


\section{ZX-calculus}

The ZX-calculus is a compact closed category $\mathfrak{C}$. The objects of $\mathfrak{C}$ are natural numbers: $0, 1, 2,  \cdots$; the tensor of objects is just addition of numbers: $m \otimes n = m+n$. The morphisms of $\mathfrak{C}$ are diagrams of the ZX-calculus. A general diagram  $D:k\to l$   with $k$ inputs and $l$ outputs is generated by:
\begin{center} 
\begin{tabular}{|r@{~}r@{~}c@{~}c|r@{~}r@{~}c@{~}c|}
\hline
$R_Z^{(n,m)}$&$:$&$n\to m$ & %
\beginpgfgraphicnamed{diagrams//generator_spider}
\InputIfFileExists{diagrams//generator_spider.tikz}{}{\input{./figures/diagrams//generator_spider.tikz}}
\endpgfgraphicnamed & $A$&$:$&$ 1\to 1$& %
\beginpgfgraphicnamed{diagrams//alphagate}
\begin{tikzpicture}
	\begin{pgfonlayer}{nodelayer}
		\node [style=none] (0) at (0, -0.5) {};
		\node [style=none] (1) at (0, 0.5) {};
		\node [style=gn] (2) at (0, 0) {$\alpha$};
	\end{pgfonlayer}
	\begin{pgfonlayer}{edgelayer}
		\draw (1.center) to (0.center);
	\end{pgfonlayer}
\end{tikzpicture}}
\endpgfgraphicnamed\\
\hline
$H$&$:$&$1\to 1$ &%
\beginpgfgraphicnamed{diagrams//HadaDecomSingleslt}
\begin{tikzpicture}
	\begin{pgfonlayer}{nodelayer}
		\node [style=H box] (0) at (-0.75, 0) {$H$};
		\node [style=none] (1) at (-0.75, -0.5) {};
		\node [style=none] (2) at (-0.75, 0.5) {};
	\end{pgfonlayer}
	\begin{pgfonlayer}{edgelayer}
		\draw (2.center) to (0);
		\draw (1.center) to (0);
	\end{pgfonlayer}
\end{tikzpicture}}
\endpgfgraphicnamed
 &  $\sigma$&$:$&$ 2\to 2$& %
\beginpgfgraphicnamed{diagrams//swap}
\InputIfFileExists{diagrams//swap.tikz}{}{\input{./figures/diagrams//swap.tikz}}
\endpgfgraphicnamed\\\hline
   $\mathbb I$&$:$&$1\to 1$&%
\beginpgfgraphicnamed{diagrams//Id}
\begin{tikzpicture}
	\begin{pgfonlayer}{nodelayer}
		\node [style=none] (1) at (0.5, 0.3) {};
		\node [style=none] (2) at (0.5, -0.3) {};
		\node [style=none] (3) at (0.5, -0.5) {};
		\node [style=none] (4) at (0.5, 0.5) {};
	\end{pgfonlayer}
	\begin{pgfonlayer}{edgelayer}
		\draw (1.center) to (2.center);
	\end{pgfonlayer}
\end{tikzpicture}}
\endpgfgraphicnamed & $e $&$:$&$0 \to 0$& %
\beginpgfgraphicnamed{diagrams//emptysquare}
\InputIfFileExists{diagrams//emptysquare.tikz}{}{\input{./figures/diagrams//emptysquare.tikz}}
\endpgfgraphicnamed\\\hline
   $C_a$&$:$&$ 0\to 2$& %
\beginpgfgraphicnamed{diagrams//cap}
\begin{tikzpicture}
	\begin{pgfonlayer}{nodelayer}
		\node [style=none] (0) at (0, -0) {};
		\node [style=none] (1) at (1, -0) {};
	\end{pgfonlayer}
	\begin{pgfonlayer}{edgelayer}
		\draw [bend left=90, looseness=1.50] (0.center) to (1.center);
	\end{pgfonlayer}
\end{tikzpicture}}
\endpgfgraphicnamed &$ C_u$&$:$&$ 2\to 0$&%
\beginpgfgraphicnamed{diagrams//cup}
\begin{tikzpicture}
	\begin{pgfonlayer}{nodelayer}
		\node [style=none] (0) at (0, 0.5) {};
		\node [style=none] (1) at (1, 0.5) {};
	\end{pgfonlayer}
	\begin{pgfonlayer}{edgelayer}
		\draw [bend right=90, looseness=1.50] (0.center) to (1.center);
	\end{pgfonlayer}
\end{tikzpicture}}
\endpgfgraphicnamed \\\hline
\end{tabular}
\end{center}
where $m,n\in \mathbb N$, $\alpha \in [0,  2\pi)$, and $e$ represents an empty diagram. For the purposes of this paper we extend the language with two new symbols (although in principle they could be eliminated, see 
lemma \ref{lem:lamb_tri_decomposition}):

\begin{center} 
	\begin{tabular}{|r@{~}r@{~}c@{~}c|r@{~}r@{~}c@{~}c|}
		\hline
		$L$&$:$&$1\to 1$  &%
\beginpgfgraphicnamed{diagrams//lambdabox}
\begin{tikzpicture}
	\begin{pgfonlayer}{nodelayer}
		\node [style=H box] (0) at (0, 0) {$\lambda$};
		\node [style=none] (1) at (0, -0.5) {};
		\node [style=none] (2) at (0, 0.5) {};
	\end{pgfonlayer}
	\begin{pgfonlayer}{edgelayer}
		\draw (2.center) to (0);
		\draw (1.center) to (0);
	\end{pgfonlayer}
\end{tikzpicture}}
\endpgfgraphicnamed &$T$&$:$&$1\to 1$&%
\beginpgfgraphicnamed{diagrams//triangle}
\begin{tikzpicture}
	\begin{pgfonlayer}{nodelayer}
		\node [style=none] (0) at (0, 0.5) {};
		\node [style=triangle] (1) at (0, 0) {};
		\node [style=none] (2) at (0, -0.5) {};
	\end{pgfonlayer}
	\begin{pgfonlayer}{edgelayer}
		\draw (0.center) to (2.center);
	\end{pgfonlayer}
\end{tikzpicture}}
\endpgfgraphicnamed \\\hline
	\end{tabular}
\end{center}
where $ \lambda  \geq 0$.

The composition of morphisms is  to combine these components in the following two ways: for any two morphisms $D_1:a\to b$ and $D_2: c\to d$, a \textit{ paralell composition} $D_1\otimes D_2 : a+c\to b+d$ is obtained by placing $D_1$ and $D_2$ side-by-side with $D_1$ on the left of $D_2$; for any two morphisms $D_1:a\to b$ and $D_2: b\to c$,  a  \textit{ sequential  composition} $D_2\circ D_1 : a\to c$ is obtained by placing $D_1$ above $D_2$, connecting the outputs of $D_1$ to the inputs of $D_2$.

There are two kinds of rules for the morphisms of $\mathfrak{C}$:  the structure rules for $\mathfrak{C}$ as an compact closed category, as well as original rewriting rules listed in Figure \ref{figure1} and our extended rules listed in Figure \ref{figure0} and Figure \ref{figure2}.

Note that all the diagrams should be read from top to bottom.

\begin{figure}[!h]
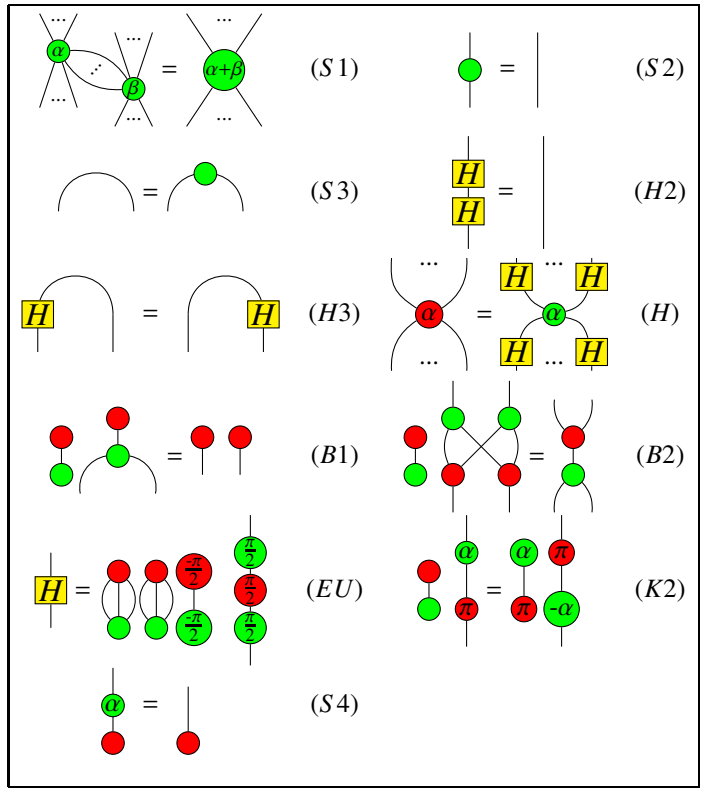

\begin{center}
\[
\quad \qquad\begin{array}{|cccc|}
\hline
\beginpgfgraphicnamed{diagrams//spiderlt}
\InputIfFileExists{diagrams//spiderlt.tikz}{}{\input{./figures/diagrams//spiderlt.tikz}}
\endpgfgraphicnamed=%
\beginpgfgraphicnamed{diagrams//spiderrt}
\InputIfFileExists{diagrams//spiderrt.tikz}{}{\input{./figures/diagrams//spiderrt.tikz}}
\endpgfgraphicnamed &(S1) &%
\beginpgfgraphicnamed{diagrams//s2new}
\InputIfFileExists{diagrams//s2new.tikz}{}{\input{./figures/diagrams//s2new.tikz}}
\endpgfgraphicnamed &(S2)\\
\beginpgfgraphicnamed{diagrams//induced_compact_structure-2wirelt}
\InputIfFileExists{diagrams//induced_compact_structure-2wirelt.tikz}{}{\input{./figures/diagrams//induced_compact_structure-2wirelt.tikz}}
\endpgfgraphicnamed=%
\beginpgfgraphicnamed{diagrams//induced_compact_structure-2wirert}
\InputIfFileExists{diagrams//induced_compact_structure-2wirert.tikz}{}{\input{./figures/diagrams//induced_compact_structure-2wirert.tikz}}
\endpgfgraphicnamed&(S3) & %
\beginpgfgraphicnamed{diagrams//hsquare}
\InputIfFileExists{diagrams//hsquare.tikz}{}{\input{./figures/diagrams//hsquare.tikz}}
\endpgfgraphicnamed &(H2)\\
\beginpgfgraphicnamed{diagrams//hslidecap}
\InputIfFileExists{diagrams//hslidecap.tikz}{}{\input{./figures/diagrams//hslidecap.tikz}}
\endpgfgraphicnamed &(H3) &%
\beginpgfgraphicnamed{diagrams//h2newlt}
\InputIfFileExists{diagrams//h2newlt.tikz}{}{\input{./figures/diagrams//h2newlt.tikz}}
\endpgfgraphicnamed=%
\beginpgfgraphicnamed{diagrams//h2newrt}
\InputIfFileExists{diagrams//h2newrt.tikz}{}{\input{./figures/diagrams//h2newrt.tikz}}
\endpgfgraphicnamed&(H)\\
\beginpgfgraphicnamed{diagrams//b1slt}
\InputIfFileExists{diagrams//b1slt.tikz}{}{\input{./figures/diagrams//b1slt.tikz}}
\endpgfgraphicnamed=%
\beginpgfgraphicnamed{diagrams//b1srt}
\InputIfFileExists{diagrams//b1srt.tikz}{}{\input{./figures/diagrams//b1srt.tikz}}
\endpgfgraphicnamed&(B1) & %
\beginpgfgraphicnamed{diagrams//b2slt}
\InputIfFileExists{diagrams//b2slt.tikz}{}{\input{./figures/diagrams//b2slt.tikz}}
\endpgfgraphicnamed=%
\beginpgfgraphicnamed{diagrams//b2srt}
\InputIfFileExists{diagrams//b2srt.tikz}{}{\input{./figures/diagrams//b2srt.tikz}}
\endpgfgraphicnamed&(B2)\\
\beginpgfgraphicnamed{diagrams//HadaDecomSingleslt}
\InputIfFileExists{diagrams//HadaDecomSingleslt.tikz}{}{\input{./figures/diagrams//HadaDecomSingleslt.tikz}}
\endpgfgraphicnamed= %
\beginpgfgraphicnamed{diagrams//HadaDecomSinglesrt}
\InputIfFileExists{diagrams//HadaDecomSinglesrt.tikz}{}{\input{./figures/diagrams//HadaDecomSinglesrt.tikz}}
\endpgfgraphicnamed&(EU)    & %
\beginpgfgraphicnamed{diagrams//k2slt}
\InputIfFileExists{diagrams//k2slt.tikz}{}{\input{./figures/diagrams//k2slt.tikz}}
\endpgfgraphicnamed=%
\beginpgfgraphicnamed{diagrams//k2srt}
\InputIfFileExists{diagrams//k2srt.tikz}{}{\input{./figures/diagrams//k2srt.tikz}}
\endpgfgraphicnamed&(K2)\\

\beginpgfgraphicnamed{diagrams//alphadelete}
\InputIfFileExists{diagrams//alphadelete.tikz}{}{\input{./figures/diagrams//alphadelete.tikz}}
\endpgfgraphicnamed&(S4) & &\\
&&&\\ 
\hline
\end{array}\]
\end{center}
  \caption{Original ZX-calculus rules, where $\alpha, \beta\in [0,~2\pi)$.}\label{figure1}  
  \end{figure}
  \begin{figure}[!h]
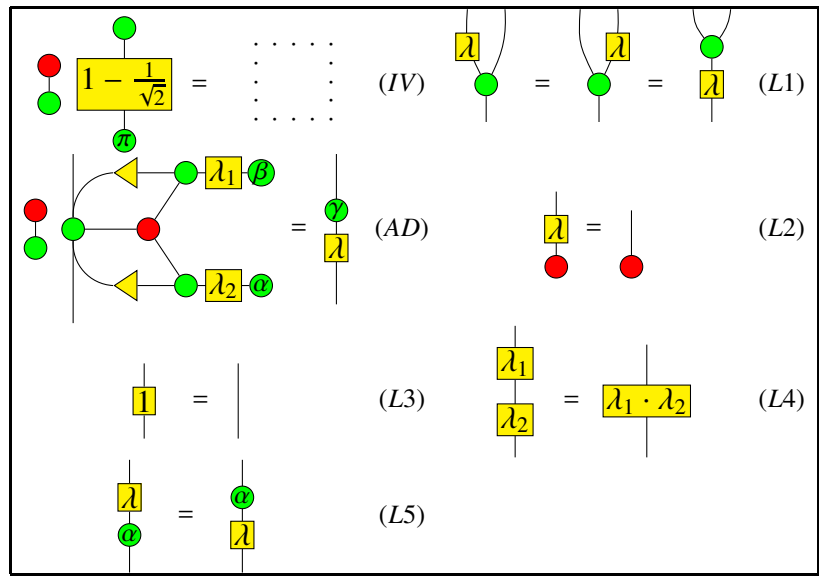

\begin{center}
\[
\quad \qquad\begin{array}{|cccc|}
\hline
\beginpgfgraphicnamed{diagrams//emptyrule}
\InputIfFileExists{diagrams//emptyrule.tikz}{}{\input{./figures/diagrams//emptyrule.tikz}}
\endpgfgraphicnamed &(IV) &%
\beginpgfgraphicnamed{diagrams//lambbranch}
\InputIfFileExists{diagrams//lambbranch.tikz}{}{\input{./figures/diagrams//lambbranch.tikz}}
\endpgfgraphicnamed &(L1)\\

\beginpgfgraphicnamed{diagrams//plus}
\InputIfFileExists{diagrams//plus.tikz}{}{\input{./figures/diagrams//plus.tikz}}
\endpgfgraphicnamed&(AD) &%
\beginpgfgraphicnamed{diagrams//lambdadelete}
\InputIfFileExists{diagrams//lambdadelete.tikz}{}{\input{./figures/diagrams//lambdadelete.tikz}}
\endpgfgraphicnamed &(L2)\\

\beginpgfgraphicnamed{diagrams//sqr1is1}
\InputIfFileExists{diagrams//sqr1is1.tikz}{}{\input{./figures/diagrams//sqr1is1.tikz}}
\endpgfgraphicnamed&(L3) &%
\beginpgfgraphicnamed{diagrams//lambdatimes}
\InputIfFileExists{diagrams//lambdatimes.tikz}{}{\input{./figures/diagrams//lambdatimes.tikz}}
\endpgfgraphicnamed&(L4)\\

\beginpgfgraphicnamed{diagrams//lambdaalpha}
\InputIfFileExists{diagrams//lambdaalpha.tikz}{}{\input{./figures/diagrams//lambdaalpha.tikz}}
\endpgfgraphicnamed&(L5) &&\\
\hline
\end{array}\]
\end{center}
  \caption{Extended ZX-calculus rules for $\lambda$ and addition, where $\lambda, \lambda_1,  \lambda_2 \geq 0, \alpha, \beta, \gamma \in [0,~2\pi);$ in (AD), $\lambda e^{i\gamma} 
  =\lambda_1 e^{i\beta}+ \lambda_2 e^{i\alpha}$.}\label{figure0}  
  \end{figure}


\begin{figure}[!h]
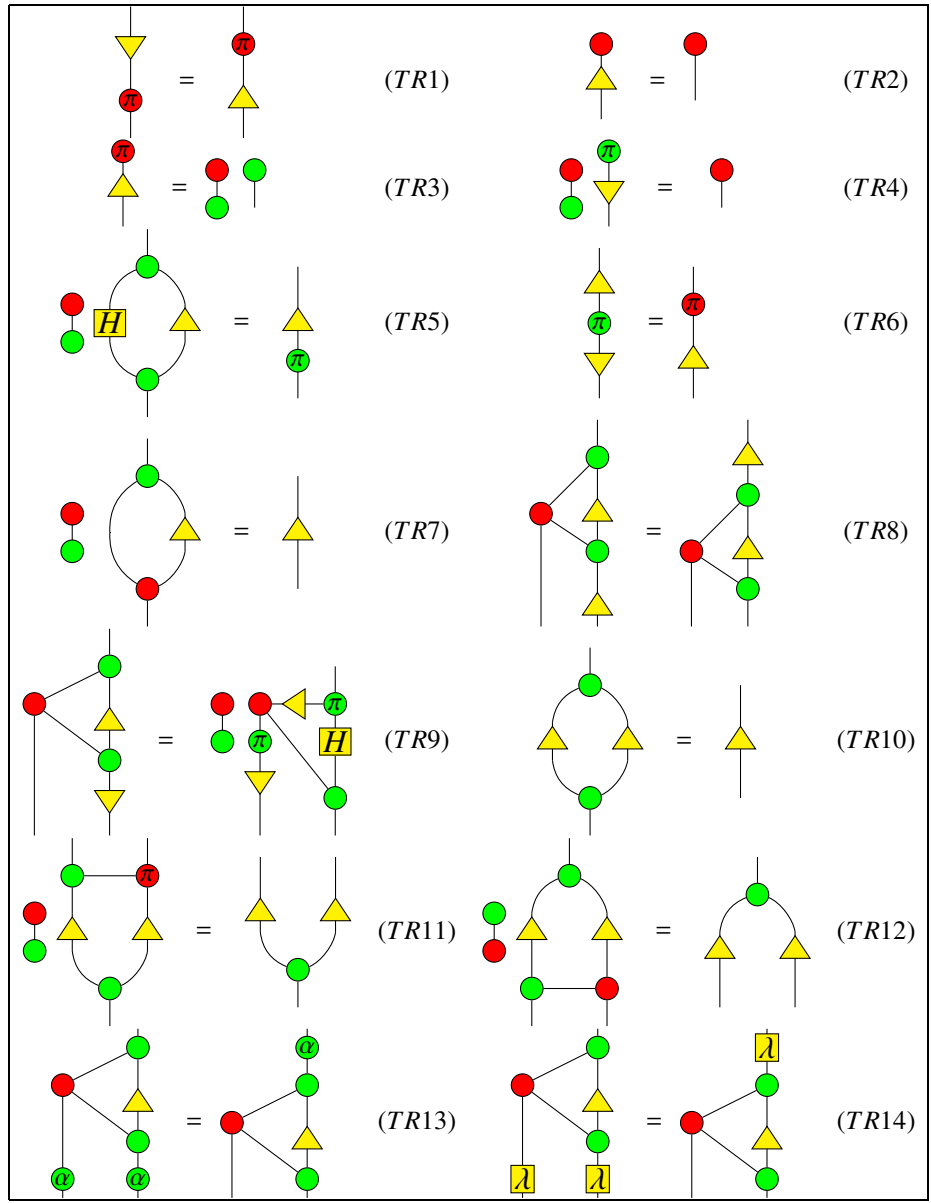

\begin{center}
\[
\quad \qquad\begin{array}{|cccc|}
\hline
 %
\beginpgfgraphicnamed{diagrams//trianglepicommute}
\InputIfFileExists{diagrams//trianglepicommute.tikz}{}{\input{./figures/diagrams//trianglepicommute.tikz}}
\endpgfgraphicnamed &(TR1) &%
\beginpgfgraphicnamed{diagrams//triangleocopy}
\InputIfFileExists{diagrams//triangleocopy.tikz}{}{\input{./figures/diagrams//triangleocopy.tikz}}
\endpgfgraphicnamed &(TR2)\\

\beginpgfgraphicnamed{diagrams//trianglepicopy}
\InputIfFileExists{diagrams//trianglepicopy.tikz}{}{\input{./figures/diagrams//trianglepicopy.tikz}}
\endpgfgraphicnamed&(TR3) & %
\beginpgfgraphicnamed{diagrams//trianglegdpicopy}
\InputIfFileExists{diagrams//trianglegdpicopy.tikz}{}{\input{./figures/diagrams//trianglegdpicopy.tikz}}
\endpgfgraphicnamed &(TR4)\\

\beginpgfgraphicnamed{diagrams//trianglehhopf}
\InputIfFileExists{diagrams//trianglehhopf.tikz}{}{\input{./figures/diagrams//trianglehhopf.tikz}}
\endpgfgraphicnamed &(TR5) &%
\beginpgfgraphicnamed{diagrams//gpiintriangles}
\InputIfFileExists{diagrams//gpiintriangles.tikz}{}{\input{./figures/diagrams//gpiintriangles.tikz}}
\endpgfgraphicnamed&(TR6)\\

\beginpgfgraphicnamed{diagrams//trianglehopf}
\InputIfFileExists{diagrams//trianglehopf.tikz}{}{\input{./figures/diagrams//trianglehopf.tikz}}
\endpgfgraphicnamed&(TR7) & %
\beginpgfgraphicnamed{diagrams//2triangleup}
\InputIfFileExists{diagrams//2triangleup.tikz}{}{\input{./figures/diagrams//2triangleup.tikz}}
\endpgfgraphicnamed&(TR8)\\

\beginpgfgraphicnamed{diagrams//2triangledown}
\InputIfFileExists{diagrams//2triangledown.tikz}{}{\input{./figures/diagrams//2triangledown.tikz}}
\endpgfgraphicnamed&(TR9)    & %
\beginpgfgraphicnamed{diagrams//2trianglehopf}
\InputIfFileExists{diagrams//2trianglehopf.tikz}{}{\input{./figures/diagrams//2trianglehopf.tikz}}
\endpgfgraphicnamed&(TR10)\\

\beginpgfgraphicnamed{diagrams//2triangledeloop}
\InputIfFileExists{diagrams//2triangledeloop.tikz}{}{\input{./figures/diagrams//2triangledeloop.tikz}}
\endpgfgraphicnamed &(TR11) &%
\beginpgfgraphicnamed{diagrams//2triangledeloopnopi}
\InputIfFileExists{diagrams//2triangledeloopnopi.tikz}{}{\input{./figures/diagrams//2triangledeloopnopi.tikz}}
\endpgfgraphicnamed &(TR12)\\

\beginpgfgraphicnamed{diagrams//alphacopyw}
\InputIfFileExists{diagrams//alphacopyw.tikz}{}{\input{./figures/diagrams//alphacopyw.tikz}}
\endpgfgraphicnamed&(TR13) &%
\beginpgfgraphicnamed{diagrams//lambdacopyw}
\InputIfFileExists{diagrams//lambdacopyw.tikz}{}{\input{./figures/diagrams//lambdacopyw.tikz}}
\endpgfgraphicnamed &(TR14)\\

\hline
\end{array}\]
\end{center}

  \caption{Extended ZX-calculus rules for triangle, where $\lambda  \geq 0, \alpha \in [0,~2\pi).$}\label{figure2}
\end{figure}

\FloatBarrier

\begin{lemma}\label{lem:lamb_tri_decomposition}
The triangle %
\beginpgfgraphicnamed{diagrams//triangle}
\begin{tikzpicture}
	\begin{pgfonlayer}{nodelayer}
		\node [style=none] (0) at (0, 0.5) {};
		\node [style=triangle] (1) at (0, 0) {};
		\node [style=none] (2) at (0, -0.5) {};
	\end{pgfonlayer}
	\begin{pgfonlayer}{edgelayer}
		\draw (0.center) to (2.center);
	\end{pgfonlayer}
\end{tikzpicture}}
\endpgfgraphicnamed and the lambda box %
\beginpgfgraphicnamed{diagrams//lambdabox}
\begin{tikzpicture}
	\begin{pgfonlayer}{nodelayer}
		\node [style=H box] (0) at (0, 0) {$\lambda$};
		\node [style=none] (1) at (0, -0.5) {};
		\node [style=none] (2) at (0, 0.5) {};
	\end{pgfonlayer}
	\begin{pgfonlayer}{edgelayer}
		\draw (2.center) to (0);
		\draw (1.center) to (0);
	\end{pgfonlayer}
\end{tikzpicture}}
\endpgfgraphicnamed  are expressible in Z and X phases.
\end{lemma}

\begin{proof}
The triangle %
\beginpgfgraphicnamed{diagrams//triangle}
\begin{tikzpicture}
	\begin{pgfonlayer}{nodelayer}
		\node [style=none] (0) at (0, 0.5) {};
		\node [style=triangle] (1) at (0, 0) {};
		\node [style=none] (2) at (0, -0.5) {};
	\end{pgfonlayer}
	\begin{pgfonlayer}{edgelayer}
		\draw (0.center) to (2.center);
	\end{pgfonlayer}
\end{tikzpicture}}
\endpgfgraphicnamed has been represented by ZX phases  in \cite{Emmanuel}. So we only need to deal with the lambda box. 
First we can write $\lambda$ as a sum of its integer part and remainder part: $\lambda= [\lambda] +\{\lambda\}$, where $ [\lambda]$ is a non-negative integer and $0\leq\{\lambda\}<1$.
Let $n= [\lambda]$.  If $n=1$, then by rule (L3), 
$$%
\beginpgfgraphicnamed{diagrams//sqr1is1}
\InputIfFileExists{diagrams//sqr1is1.tikz}{}{\input{./figures/diagrams//sqr1is1.tikz}}
\endpgfgraphicnamed $$
If  $n=0$, then 
$$%
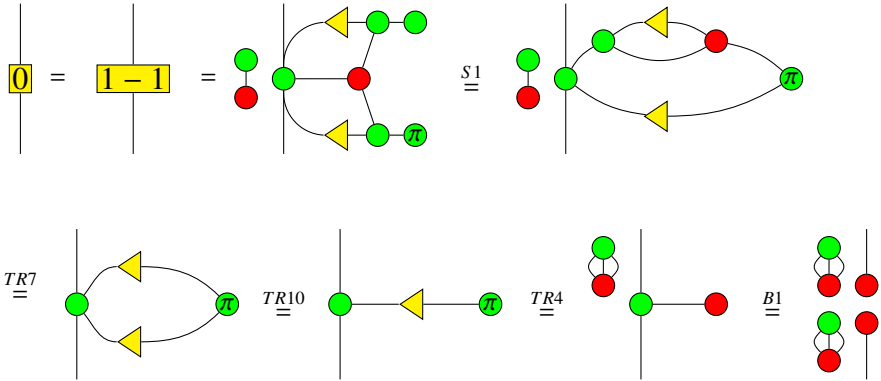
\beginpgfgraphicnamed{diagrams//lemma0}
\InputIfFileExists{diagrams//lemma0.tikz}{}{\input{./figures/diagrams//lemma0.tikz}}
\endpgfgraphicnamed $$
 
 If $n\geq 2$, then by rule (AD) and induction, we have 
 $$%
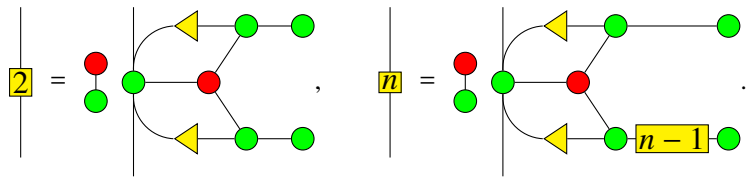
\beginpgfgraphicnamed{diagrams//lexpress1}
\InputIfFileExists{diagrams//lexpress1.tikz}{}{\input{./figures/diagrams//lexpress1.tikz}}
\endpgfgraphicnamed, \quad\quad %
\beginpgfgraphicnamed{diagrams//lexpress2}
\InputIfFileExists{diagrams//lexpress2.tikz}{}{\input{./figures/diagrams//lexpress2.tikz}}
\endpgfgraphicnamed.$$
 Again by rule (AD),   we have
  $$%
\beginpgfgraphicnamed{diagrams//lexpress3}
\InputIfFileExists{diagrams//lexpress3.tikz}{}{\input{./figures/diagrams//lexpress3.tikz}}
\endpgfgraphicnamed, \quad\quad where~ \alpha=arccos\frac{\{\lambda\}}{2}.$$
  Therefore, we have
    $$%
\beginpgfgraphicnamed{diagrams//lexpress4}
\InputIfFileExists{diagrams//lexpress4.tikz}{}{\input{./figures/diagrams//lexpress4.tikz}}
\endpgfgraphicnamed.$$

\end{proof}

The diagrams in the ZX-calculus have a standard interpretation $\llbracket \cdot \rrbracket$ in the category of Hilbert spaces:
\[
\left\llbracket %
\beginpgfgraphicnamed{diagrams//generator_spider}
\InputIfFileExists{diagrams//generator_spider.tikz}{}{\input{./figures/diagrams//generator_spider.tikz}}
\endpgfgraphicnamed \right\rrbracket=\ket{0}^{\otimes m}\bra{0}^{\otimes n}+\ket{1}^{\otimes m}\bra{1}^{\otimes n},
\quad
\left\llbracket%
\beginpgfgraphicnamed{diagrams//alphagate}
\begin{tikzpicture}
	\begin{pgfonlayer}{nodelayer}
		\node [style=none] (0) at (0, -0.5) {};
		\node [style=none] (1) at (0, 0.5) {};
		\node [style=gn] (2) at (0, 0) {$\alpha$};
	\end{pgfonlayer}
	\begin{pgfonlayer}{edgelayer}
		\draw (1.center) to (0.center);
	\end{pgfonlayer}
\end{tikzpicture}}
\endpgfgraphicnamed\right\rrbracket=\ket{0}\bra{0}+e^{i\alpha}\ket{1}\bra{1}=\begin{pmatrix}
        1 & 0 \\
        0 & e^{i\alpha}
 \end{pmatrix}.
\]

\[
\left\llbracket%
\beginpgfgraphicnamed{diagrams//HadaDecomSingleslt}
\begin{tikzpicture}
	\begin{pgfonlayer}{nodelayer}
		\node [style=H box] (0) at (-0.75, 0) {$H$};
		\node [style=none] (1) at (-0.75, -0.5) {};
		\node [style=none] (2) at (-0.75, 0.5) {};
	\end{pgfonlayer}
	\begin{pgfonlayer}{edgelayer}
		\draw (2.center) to (0);
		\draw (1.center) to (0);
	\end{pgfonlayer}
\end{tikzpicture}}
\endpgfgraphicnamed\right\rrbracket=\frac{1}{\sqrt{2}}\begin{pmatrix}
        1 & 1 \\
        1 & -1
 \end{pmatrix}, \quad
  \left\llbracket%
\beginpgfgraphicnamed{diagrams//lambdabox}
\begin{tikzpicture}
	\begin{pgfonlayer}{nodelayer}
		\node [style=H box] (0) at (0, 0) {$\lambda$};
		\node [style=none] (1) at (0, -0.5) {};
		\node [style=none] (2) at (0, 0.5) {};
	\end{pgfonlayer}
	\begin{pgfonlayer}{edgelayer}
		\draw (2.center) to (0);
		\draw (1.center) to (0);
	\end{pgfonlayer}
\end{tikzpicture}}
\endpgfgraphicnamed\right\rrbracket=\begin{pmatrix}
        1 & 0 \\
        0 & \lambda
 \end{pmatrix}, \quad
  \left\llbracket%
\beginpgfgraphicnamed{diagrams//triangle}
\begin{tikzpicture}
	\begin{pgfonlayer}{nodelayer}
		\node [style=none] (0) at (0, 0.5) {};
		\node [style=triangle] (1) at (0, 0) {};
		\node [style=none] (2) at (0, -0.5) {};
	\end{pgfonlayer}
	\begin{pgfonlayer}{edgelayer}
		\draw (0.center) to (2.center);
	\end{pgfonlayer}
\end{tikzpicture}}
\endpgfgraphicnamed\right\rrbracket=\begin{pmatrix}
        1 & 1 \\
        0 & 1
 \end{pmatrix}, \quad
   \left\llbracket%
\beginpgfgraphicnamed{diagrams//emptysquare}
\InputIfFileExists{diagrams//emptysquare.tikz}{}{\input{./figures/diagrams//emptysquare.tikz}}
\endpgfgraphicnamed\right\rrbracket=1.
   \]

\[
\left\llbracket%
\beginpgfgraphicnamed{diagrams//Id}
\begin{tikzpicture}
	\begin{pgfonlayer}{nodelayer}
		\node [style=none] (1) at (0.5, 0.3) {};
		\node [style=none] (2) at (0.5, -0.3) {};
		\node [style=none] (3) at (0.5, -0.5) {};
		\node [style=none] (4) at (0.5, 0.5) {};
	\end{pgfonlayer}
	\begin{pgfonlayer}{edgelayer}
		\draw (1.center) to (2.center);
	\end{pgfonlayer}
\end{tikzpicture}}
\endpgfgraphicnamed\right\rrbracket=\begin{pmatrix}
        1 & 0 \\
        0 & 1
 \end{pmatrix}, \quad
  \left\llbracket%
\beginpgfgraphicnamed{diagrams//swap}
\InputIfFileExists{diagrams//swap.tikz}{}{\input{./figures/diagrams//swap.tikz}}
\endpgfgraphicnamed\right\rrbracket=\begin{pmatrix}
        1 & 0 & 0 & 0 \\
        0 & 0 & 1 & 0 \\
        0 & 1 & 0 & 0 \\
        0 & 0 & 0 & 1 
 \end{pmatrix}, \quad
  \left\llbracket%
\beginpgfgraphicnamed{diagrams//cap}
\begin{tikzpicture}
	\begin{pgfonlayer}{nodelayer}
		\node [style=none] (0) at (0, -0) {};
		\node [style=none] (1) at (1, -0) {};
	\end{pgfonlayer}
	\begin{pgfonlayer}{edgelayer}
		\draw [bend left=90, looseness=1.50] (0.center) to (1.center);
	\end{pgfonlayer}
\end{tikzpicture}}
\endpgfgraphicnamed\right\rrbracket=\begin{pmatrix}
        1  \\
        0  \\
        0  \\
        1  \\
 \end{pmatrix}, \quad
   \left\llbracket%
\beginpgfgraphicnamed{diagrams//cup}
\begin{tikzpicture}
	\begin{pgfonlayer}{nodelayer}
		\node [style=none] (0) at (0, 0.5) {};
		\node [style=none] (1) at (1, 0.5) {};
	\end{pgfonlayer}
	\begin{pgfonlayer}{edgelayer}
		\draw [bend right=90, looseness=1.50] (0.center) to (1.center);
	\end{pgfonlayer}
\end{tikzpicture}}
\endpgfgraphicnamed\right\rrbracket=\begin{pmatrix}
        1 & 0 & 0 & 1 
         \end{pmatrix}.
   \]

\[  \llbracket D_1\otimes D_2  \rrbracket =  \llbracket D_1  \rrbracket \otimes  \llbracket  D_2  \rrbracket, \quad 
 \llbracket D_1\circ D_2  \rrbracket =  \llbracket D_1  \rrbracket \circ  \llbracket  D_2  \rrbracket.
  \]

It can be verified that the interpretation $\llbracket \cdot \rrbracket$ is a monoidal functor.

\section{ZW-calculus}
The ZW-calculus is a compact closed category $\mathfrak{F}$. The objects of $\mathfrak{F}$ are natural numbers: $0, 1, 2,  \cdots$; the tensor of 
objects is just addition of numbers: $m \otimes n = m+n$. The morphisms of $\mathfrak{F}$ are diagrams of the ZX-calculus. A general diagram  $D:k\to l$   with $k$ inputs and $l$ outputs is generated by:
\begin{center} 
\begin{tabular}{|r@{~}r@{~}c@{~}c|r@{~}r@{~}c@{~}c|}
\hline
$Z^{(n,m)}$&$:$&$n\to m$ & %
\beginpgfgraphicnamed{diagrams//generatorwtspider}
\InputIfFileExists{diagrams//generatorwtspider.tikz}{}{\input{./figures/diagrams//generatorwtspider.tikz}}
\endpgfgraphicnamed & $R$&$:$&$ 1\to 1$& %
\beginpgfgraphicnamed{diagrams//rgatewhite}
\begin{tikzpicture}
	\begin{pgfonlayer}{nodelayer}
		\node [style=wn] (0) at (0, 0) {};
		\node [style=none] (1) at (0, -0.5) {};
		\node [style=none] (2) at (0.25, 0) {$r$};
		\node [style=none] (3) at (0, 0.5) {};
	\end{pgfonlayer}
	\begin{pgfonlayer}{edgelayer}
		\draw (3.center) to (1.center);
	\end{pgfonlayer}
\end{tikzpicture}}
\endpgfgraphicnamed\\
\hline
$\tau$&$:$&$2\to 2$ &%
\beginpgfgraphicnamed{diagrams//corsszw}
\InputIfFileExists{diagrams//corsszw.tikz}{}{\input{./figures/diagrams//corsszw.tikz}}
\endpgfgraphicnamed
 & $P$&$:$&$1\to 1$  &%
\beginpgfgraphicnamed{diagrams//piblack}
\begin{tikzpicture}
	\begin{pgfonlayer}{nodelayer}
		\node [style=bn] (0) at (0, 0) {};
		\node [style=none] (1) at (0, -0.5) {};
		\node [style=none] (2) at (0, 0.5) {};
	\end{pgfonlayer}
	\begin{pgfonlayer}{edgelayer}
		\draw (2.center) to (1.center);
	\end{pgfonlayer}
\end{tikzpicture}}
\endpgfgraphicnamed \\\hline
  $\sigma$&$:$&$ 2\to 2$& %
\beginpgfgraphicnamed{diagrams//swap}
\InputIfFileExists{diagrams//swap.tikz}{}{\input{./figures/diagrams//swap.tikz}}
\endpgfgraphicnamed &$\mathbb I$&$:$&$1\to 1$&%
\beginpgfgraphicnamed{diagrams//Id}
\begin{tikzpicture}
	\begin{pgfonlayer}{nodelayer}
		\node [style=none] (1) at (0.5, 0.3) {};
		\node [style=none] (2) at (0.5, -0.3) {};
		\node [style=none] (3) at (0.5, -0.5) {};
		\node [style=none] (4) at (0.5, 0.5) {};
	\end{pgfonlayer}
	\begin{pgfonlayer}{edgelayer}
		\draw (1.center) to (2.center);
	\end{pgfonlayer}
\end{tikzpicture}}
\endpgfgraphicnamed \\\hline
   $e $&$:$&$0 \to 0$& %
\beginpgfgraphicnamed{diagrams//emptysquare}
\InputIfFileExists{diagrams//emptysquare.tikz}{}{\input{./figures/diagrams//emptysquare.tikz}}
\endpgfgraphicnamed &$W$&$:$&$1\to 2$&%
\beginpgfgraphicnamed{diagrams//wblack}
\begin{tikzpicture}
	\begin{pgfonlayer}{nodelayer}
		\node [style=bn] (0) at (0, 0) {};
		\node [style=none] (1) at (-0.25, -0.5) {};
		\node [style=none] (2) at (0.25, -0.5) {};
		\node [style=none] (3) at (0, 0.5) {};
	\end{pgfonlayer}
	\begin{pgfonlayer}{edgelayer}
		\draw (3.center) to (0);
		\draw (0) to (1.center);
		\draw (0) to (2.center);
	\end{pgfonlayer}
\end{tikzpicture}}
\endpgfgraphicnamed 
  \\\hline
   $C_a$&$:$&$ 0\to 2$& %
\beginpgfgraphicnamed{diagrams//cap}
\begin{tikzpicture}
	\begin{pgfonlayer}{nodelayer}
		\node [style=none] (0) at (0, -0) {};
		\node [style=none] (1) at (1, -0) {};
	\end{pgfonlayer}
	\begin{pgfonlayer}{edgelayer}
		\draw [bend left=90, looseness=1.50] (0.center) to (1.center);
	\end{pgfonlayer}
\end{tikzpicture}}
\endpgfgraphicnamed &$ C_u$&$:$&$ 2\to 0$&%
\beginpgfgraphicnamed{diagrams//cup}
\begin{tikzpicture}
	\begin{pgfonlayer}{nodelayer}
		\node [style=none] (0) at (0, 0.5) {};
		\node [style=none] (1) at (1, 0.5) {};
	\end{pgfonlayer}
	\begin{pgfonlayer}{edgelayer}
		\draw [bend right=90, looseness=1.50] (0.center) to (1.center);
	\end{pgfonlayer}
\end{tikzpicture}}
\endpgfgraphicnamed \\\hline
\end{tabular}
\end{center}
where $m,n\in \mathbb{N}$,   $r \in  \mathbb{C}$, and $e$ represents an empty diagram.

The composition of morphisms is  to combine these components in the following two ways: for any two morphisms $D_1:a\to b$ and $D_2: c\to d$, a \textit{ paralell composition} $D_1\otimes D_2 : a+c\to b+d$ is obtained by placing $D_1$ and $D_2$ side-by-side with $D_1$ on the left of $D_2$;
 for any two morphisms $D_1:a\to b$ and $D_2: b\to c$,  a  \textit{ sequential  composition} $D_2\circ D_1 : a\to c$ is obtained by placing $D_1$ above $D_2$, connecting the outputs of $D_1$ to the inputs of $D_2$.

There are two kinds of rules for the morphisms of $\mathfrak{F}$:  the structure rules for $\mathfrak{F}$ as an compact closed category, as well as the rewriting rules listed in Figure \ref{figure3}, \ref{figure4}, \ref{figure5}, \ref{figure6}.

Note that all the diagrams should be read from top to bottom.

\begin{figure}[!h]
\begin{center}
\[
\quad \qquad\begin{array}{|cccc|}
\hline
\beginpgfgraphicnamed{diagrams//reix2}
\InputIfFileExists{diagrams//reix2.tikz}{}{\input{./figures/diagrams//reix2.tikz}}
\endpgfgraphicnamed& &%
\beginpgfgraphicnamed{diagrams//reix3}
\InputIfFileExists{diagrams//reix3.tikz}{}{\input{./figures/diagrams//reix3.tikz}}
\endpgfgraphicnamed & \\
\beginpgfgraphicnamed{diagrams//natnx}
\InputIfFileExists{diagrams//natnx.tikz}{}{\input{./figures/diagrams//natnx.tikz}}
\endpgfgraphicnamed & & %
\beginpgfgraphicnamed{diagrams//natex}
\InputIfFileExists{diagrams//natex.tikz}{}{\input{./figures/diagrams//natex.tikz}}
\endpgfgraphicnamed &\\
\beginpgfgraphicnamed{diagrams//reix1}
\InputIfFileExists{diagrams//reix1.tikz}{}{\input{./figures/diagrams//reix1.tikz}}
\endpgfgraphicnamed & &%
\beginpgfgraphicnamed{diagrams//uncowL}
\InputIfFileExists{diagrams//uncowL.tikz}{}{\input{./figures/diagrams//uncowL.tikz}}
\endpgfgraphicnamed&\\
\beginpgfgraphicnamed{diagrams//uncowR}
\InputIfFileExists{diagrams//uncowR.tikz}{}{\input{./figures/diagrams//uncowR.tikz}}
\endpgfgraphicnamed& & %
\beginpgfgraphicnamed{diagrams//natww}
\InputIfFileExists{diagrams//natww.tikz}{}{\input{./figures/diagrams//natww.tikz}}
\endpgfgraphicnamed&\\
\beginpgfgraphicnamed{diagrams//natwx}
\InputIfFileExists{diagrams//natwx.tikz}{}{\input{./figures/diagrams//natwx.tikz}}
\endpgfgraphicnamed&  & %
\beginpgfgraphicnamed{diagrams//comcow}
\InputIfFileExists{diagrams//comcow.tikz}{}{\input{./figures/diagrams//comcow.tikz}}
\endpgfgraphicnamed&\\
\beginpgfgraphicnamed{diagrams//natmw}
\InputIfFileExists{diagrams//natmw.tikz}{}{\input{./figures/diagrams//natmw.tikz}}
\endpgfgraphicnamed & &%
\beginpgfgraphicnamed{diagrams//natmnw}
\InputIfFileExists{diagrams//natmnw.tikz}{}{\input{./figures/diagrams//natmnw.tikz}}
\endpgfgraphicnamed &\\
\beginpgfgraphicnamed{diagrams//natmnew}
\InputIfFileExists{diagrams//natmnew.tikz}{}{\input{./figures/diagrams//natmnew.tikz}}
\endpgfgraphicnamed& &%
\beginpgfgraphicnamed{diagrams//hopf}
\InputIfFileExists{diagrams//hopf.tikz}{}{\input{./figures/diagrams//hopf.tikz}}
\endpgfgraphicnamed &\\
\beginpgfgraphicnamed{diagrams//sym3}
\InputIfFileExists{diagrams//sym3.tikz}{}{\input{./figures/diagrams//sym3.tikz}}
\endpgfgraphicnamed & &  &\\
\hline
\end{array}\]
\end{center}

  \caption{ZW-calculus rules I}\label{figure3}
\end{figure}

\begin{figure}[!h]
\begin{center}
\[
\quad \qquad\begin{array}{|cccc|}
\hline
\beginpgfgraphicnamed{diagrams//sym2}
\InputIfFileExists{diagrams//sym2.tikz}{}{\input{./figures/diagrams//sym2.tikz}}
\endpgfgraphicnamed &%
\beginpgfgraphicnamed{diagrams//inv}
\InputIfFileExists{diagrams//inv.tikz}{}{\input{./figures/diagrams//inv.tikz}}
\endpgfgraphicnamed&&\\
\beginpgfgraphicnamed{diagrams//antnx}
\InputIfFileExists{diagrams//antnx.tikz}{}{\input{./figures/diagrams//antnx.tikz}}
\endpgfgraphicnamed  &%
\beginpgfgraphicnamed{diagrams//symz}
\InputIfFileExists{diagrams//symz.tikz}{}{\input{./figures/diagrams//symz.tikz}}
\endpgfgraphicnamed&&\\
\beginpgfgraphicnamed{diagrams//uncozR}
\InputIfFileExists{diagrams//uncozR.tikz}{}{\input{./figures/diagrams//uncozR.tikz}}
\endpgfgraphicnamed  &%
\beginpgfgraphicnamed{diagrams//natzz}
\InputIfFileExists{diagrams//natzz.tikz}{}{\input{./figures/diagrams//natzz.tikz}}
\endpgfgraphicnamed  & & \\
\beginpgfgraphicnamed{diagrams//ph}
\InputIfFileExists{diagrams//ph.tikz}{}{\input{./figures/diagrams//ph.tikz}}
\endpgfgraphicnamed  &%
\beginpgfgraphicnamed{diagrams//natnc}
\InputIfFileExists{diagrams//natnc.tikz}{}{\input{./figures/diagrams//natnc.tikz}}
\endpgfgraphicnamed  & & \\
\beginpgfgraphicnamed{diagrams//natmc}
\InputIfFileExists{diagrams//natmc.tikz}{}{\input{./figures/diagrams//natmc.tikz}}
\endpgfgraphicnamed  &%
\beginpgfgraphicnamed{diagrams//loop}
\InputIfFileExists{diagrams//loop.tikz}{}{\input{./figures/diagrams//loop.tikz}}
\endpgfgraphicnamed  & & \\
\beginpgfgraphicnamed{diagrams//unx}
\InputIfFileExists{diagrams//unx.tikz}{}{\input{./figures/diagrams//unx.tikz}}
\endpgfgraphicnamed  &%
\beginpgfgraphicnamed{diagrams//rng1}
\InputIfFileExists{diagrams//rng1.tikz}{}{\input{./figures/diagrams//rng1.tikz}}
\endpgfgraphicnamed  & & \\
\beginpgfgraphicnamed{diagrams//rng-1}
\InputIfFileExists{diagrams//rng-1.tikz}{}{\input{./figures/diagrams//rng-1.tikz}}
\endpgfgraphicnamed  &%
\beginpgfgraphicnamed{diagrams//rngrsx}
\InputIfFileExists{diagrams//rngrsx.tikz}{}{\input{./figures/diagrams//rngrsx.tikz}}
\endpgfgraphicnamed  & & \\
\hline
\end{array}\]
\end{center}

  \caption{ZW-calculus rules II}\label{figure4}
\end{figure}

\begin{figure}[!h]
\begin{center}
\[
\quad \qquad\begin{array}{|cccc|}
\hline
\beginpgfgraphicnamed{diagrams//rngrsp}
\InputIfFileExists{diagrams//rngrsp.tikz}{}{\input{./figures/diagrams//rngrsp.tikz}}
\endpgfgraphicnamed  &%
\beginpgfgraphicnamed{diagrams//natrc}
\InputIfFileExists{diagrams//natrc.tikz}{}{\input{./figures/diagrams//natrc.tikz}}
\endpgfgraphicnamed   & & \\
\beginpgfgraphicnamed{diagrams//natrec}
\InputIfFileExists{diagrams//natrec.tikz}{}{\input{./figures/diagrams//natrec.tikz}}
\endpgfgraphicnamed  &%
\beginpgfgraphicnamed{diagrams//phr}
\InputIfFileExists{diagrams//phr.tikz}{}{\input{./figures/diagrams//phr.tikz}}
\endpgfgraphicnamed  & & \\
\hline
\end{array}\]
\end{center}

  \caption{ZW-calculus rules III}\label{figure5}
\end{figure}

The diagrams in the ZX-calculus have a standard interpretation $\llbracket \cdot \rrbracket$ in the category of Hilbert spaces:
\[
\left\llbracket %
\beginpgfgraphicnamed{diagrams//generatorwtspider}
\InputIfFileExists{diagrams//generatorwtspider.tikz}{}{\input{./figures/diagrams//generatorwtspider.tikz}}
\endpgfgraphicnamed \right\rrbracket=\ket{0}^{\otimes m}\bra{0}^{\otimes n}+\ket{1}^{\otimes m}\bra{1}^{\otimes n},
\quad
\left\llbracket%
\beginpgfgraphicnamed{diagrams//rgatewhite}
\begin{tikzpicture}
	\begin{pgfonlayer}{nodelayer}
		\node [style=wn] (0) at (0, 0) {};
		\node [style=none] (1) at (0, -0.5) {};
		\node [style=none] (2) at (0.25, 0) {$r$};
		\node [style=none] (3) at (0, 0.5) {};
	\end{pgfonlayer}
	\begin{pgfonlayer}{edgelayer}
		\draw (3.center) to (1.center);
	\end{pgfonlayer}
\end{tikzpicture}}
\endpgfgraphicnamed\right\rrbracket=\ket{0}\bra{0}+r\ket{1}\bra{1}=\begin{pmatrix}
        1 & 0 \\
        0 & r
 \end{pmatrix}.
\]

\[
\left\llbracket%
\beginpgfgraphicnamed{diagrams//corsszw}
\InputIfFileExists{diagrams//corsszw.tikz}{}{\input{./figures/diagrams//corsszw.tikz}}
\endpgfgraphicnamed\right\rrbracket=\begin{pmatrix}
         1 & 0 & 0 & 0 \\
        0 & 0 & 1 & 0 \\
        0 & 1 & 0 & 0 \\
        0 & 0 & 0 & -1 
 \end{pmatrix}, \quad
  \left\llbracket%
\beginpgfgraphicnamed{diagrams//wblack}
\begin{tikzpicture}
	\begin{pgfonlayer}{nodelayer}
		\node [style=bn] (0) at (0, 0) {};
		\node [style=none] (1) at (-0.25, -0.5) {};
		\node [style=none] (2) at (0.25, -0.5) {};
		\node [style=none] (3) at (0, 0.5) {};
	\end{pgfonlayer}
	\begin{pgfonlayer}{edgelayer}
		\draw (3.center) to (0);
		\draw (0) to (1.center);
		\draw (0) to (2.center);
	\end{pgfonlayer}
\end{tikzpicture}}
\endpgfgraphicnamed\right\rrbracket=\begin{pmatrix}
        0 & 1  \\
        1 & 0 \\
        1 & 0  \\
        0 & 0 
 \end{pmatrix}, \quad
  \left\llbracket%
\beginpgfgraphicnamed{diagrams//piblack}
\begin{tikzpicture}
	\begin{pgfonlayer}{nodelayer}
		\node [style=bn] (0) at (0, 0) {};
		\node [style=none] (1) at (0, -0.5) {};
		\node [style=none] (2) at (0, 0.5) {};
	\end{pgfonlayer}
	\begin{pgfonlayer}{edgelayer}
		\draw (2.center) to (1.center);
	\end{pgfonlayer}
\end{tikzpicture}}
\endpgfgraphicnamed\right\rrbracket=\begin{pmatrix}
        0 & 1 \\
        1 & 0
 \end{pmatrix}, \quad
   \left\llbracket%
\beginpgfgraphicnamed{diagrams//emptysquare}
\InputIfFileExists{diagrams//emptysquare.tikz}{}{\input{./figures/diagrams//emptysquare.tikz}}
\endpgfgraphicnamed\right\rrbracket=1.
   \]

\[
\left\llbracket%
\beginpgfgraphicnamed{diagrams//Id}
\begin{tikzpicture}
	\begin{pgfonlayer}{nodelayer}
		\node [style=none] (1) at (0.5, 0.3) {};
		\node [style=none] (2) at (0.5, -0.3) {};
		\node [style=none] (3) at (0.5, -0.5) {};
		\node [style=none] (4) at (0.5, 0.5) {};
	\end{pgfonlayer}
	\begin{pgfonlayer}{edgelayer}
		\draw (1.center) to (2.center);
	\end{pgfonlayer}
\end{tikzpicture}}
\endpgfgraphicnamed\right\rrbracket=\begin{pmatrix}
        1 & 0 \\
        0 & 1
 \end{pmatrix}, \quad
  \left\llbracket%
\beginpgfgraphicnamed{diagrams//swap}
\InputIfFileExists{diagrams//swap.tikz}{}{\input{./figures/diagrams//swap.tikz}}
\endpgfgraphicnamed\right\rrbracket=\begin{pmatrix}
        1 & 0 & 0 & 0 \\
        0 & 0 & 1 & 0 \\
        0 & 1 & 0 & 0 \\
        0 & 0 & 0 & 1 
 \end{pmatrix}, \quad
  \left\llbracket%
\beginpgfgraphicnamed{diagrams//cap}
\begin{tikzpicture}
	\begin{pgfonlayer}{nodelayer}
		\node [style=none] (0) at (0, -0) {};
		\node [style=none] (1) at (1, -0) {};
	\end{pgfonlayer}
	\begin{pgfonlayer}{edgelayer}
		\draw [bend left=90, looseness=1.50] (0.center) to (1.center);
	\end{pgfonlayer}
\end{tikzpicture}}
\endpgfgraphicnamed\right\rrbracket=\begin{pmatrix}
        1  \\
        0  \\
        0  \\
        1  \\
 \end{pmatrix}, \quad
   \left\llbracket%
\beginpgfgraphicnamed{diagrams//cup}
\begin{tikzpicture}
	\begin{pgfonlayer}{nodelayer}
		\node [style=none] (0) at (0, 0.5) {};
		\node [style=none] (1) at (1, 0.5) {};
	\end{pgfonlayer}
	\begin{pgfonlayer}{edgelayer}
		\draw [bend right=90, looseness=1.50] (0.center) to (1.center);
	\end{pgfonlayer}
\end{tikzpicture}}
\endpgfgraphicnamed\right\rrbracket=\begin{pmatrix}
        1 & 0 & 0 & 1 
         \end{pmatrix}.
   \]

\[  \llbracket D_1\otimes D_2  \rrbracket =  \llbracket D_1  \rrbracket \otimes  \llbracket  D_2  \rrbracket, \quad 
 \llbracket D_1\circ D_2  \rrbracket =  \llbracket D_1  \rrbracket \circ  \llbracket  D_2  \rrbracket.
  \]

It can be verified that the interpretation $\llbracket \cdot \rrbracket$ is a monoidal functor.

\section{Interpretations from ZX-calculus to  ZW-calculus and back forth}
First we define the interpretation $\llbracket \cdot \rrbracket_{XW}$   from ZX-calculus to  ZW-calculus  as follows:
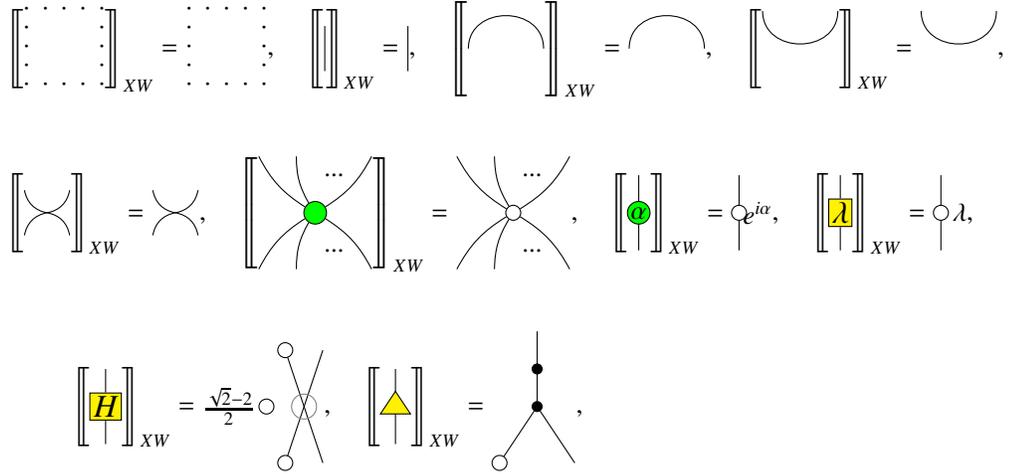
\[
 \left\llbracket%
\beginpgfgraphicnamed{diagrams//emptysquare}
\InputIfFileExists{diagrams//emptysquare.tikz}{}{\input{./figures/diagrams//emptysquare.tikz}}
\endpgfgraphicnamed\right\rrbracket_{XW}=  %
\beginpgfgraphicnamed{diagrams//emptysquare}
\InputIfFileExists{diagrams//emptysquare.tikz}{}{\input{./figures/diagrams//emptysquare.tikz}}
\endpgfgraphicnamed,  \quad
  \left\llbracket%
\beginpgfgraphicnamed{diagrams//Id}
\begin{tikzpicture}
	\begin{pgfonlayer}{nodelayer}
		\node [style=none] (1) at (0.5, 0.3) {};
		\node [style=none] (2) at (0.5, -0.3) {};
		\node [style=none] (3) at (0.5, -0.5) {};
		\node [style=none] (4) at (0.5, 0.5) {};
	\end{pgfonlayer}
	\begin{pgfonlayer}{edgelayer}
		\draw (1.center) to (2.center);
	\end{pgfonlayer}
\end{tikzpicture}}
\endpgfgraphicnamed\right\rrbracket_{XW}=  %
\beginpgfgraphicnamed{diagrams//Id}
\begin{tikzpicture}
	\begin{pgfonlayer}{nodelayer}
		\node [style=none] (1) at (0.5, 0.3) {};
		\node [style=none] (2) at (0.5, -0.3) {};
		\node [style=none] (3) at (0.5, -0.5) {};
		\node [style=none] (4) at (0.5, 0.5) {};
	\end{pgfonlayer}
	\begin{pgfonlayer}{edgelayer}
		\draw (1.center) to (2.center);
	\end{pgfonlayer}
\end{tikzpicture}}
\endpgfgraphicnamed,   \quad
 \left\llbracket%
\beginpgfgraphicnamed{diagrams//cap}
\begin{tikzpicture}
	\begin{pgfonlayer}{nodelayer}
		\node [style=none] (0) at (0, -0) {};
		\node [style=none] (1) at (1, -0) {};
	\end{pgfonlayer}
	\begin{pgfonlayer}{edgelayer}
		\draw [bend left=90, looseness=1.50] (0.center) to (1.center);
	\end{pgfonlayer}
\end{tikzpicture}}
\endpgfgraphicnamed\right\rrbracket_{XW}=  %
\beginpgfgraphicnamed{diagrams//cap}
\begin{tikzpicture}
	\begin{pgfonlayer}{nodelayer}
		\node [style=none] (0) at (0, -0) {};
		\node [style=none] (1) at (1, -0) {};
	\end{pgfonlayer}
	\begin{pgfonlayer}{edgelayer}
		\draw [bend left=90, looseness=1.50] (0.center) to (1.center);
	\end{pgfonlayer}
\end{tikzpicture}}
\endpgfgraphicnamed,   \quad
  \left\llbracket%
\beginpgfgraphicnamed{diagrams//cup}
\begin{tikzpicture}
	\begin{pgfonlayer}{nodelayer}
		\node [style=none] (0) at (0, 0.5) {};
		\node [style=none] (1) at (1, 0.5) {};
	\end{pgfonlayer}
	\begin{pgfonlayer}{edgelayer}
		\draw [bend right=90, looseness=1.50] (0.center) to (1.center);
	\end{pgfonlayer}
\end{tikzpicture}}
\endpgfgraphicnamed\right\rrbracket_{XW}=  %
\beginpgfgraphicnamed{diagrams//cup}
\begin{tikzpicture}
	\begin{pgfonlayer}{nodelayer}
		\node [style=none] (0) at (0, 0.5) {};
		\node [style=none] (1) at (1, 0.5) {};
	\end{pgfonlayer}
	\begin{pgfonlayer}{edgelayer}
		\draw [bend right=90, looseness=1.50] (0.center) to (1.center);
	\end{pgfonlayer}
\end{tikzpicture}}
\endpgfgraphicnamed,  
  \]
  
  \[
   \left\llbracket%
\beginpgfgraphicnamed{diagrams//swap}
\InputIfFileExists{diagrams//swap.tikz}{}{\input{./figures/diagrams//swap.tikz}}
\endpgfgraphicnamed\right\rrbracket_{XW}=  %
\beginpgfgraphicnamed{diagrams//swap}
\InputIfFileExists{diagrams//swap.tikz}{}{\input{./figures/diagrams//swap.tikz}}
\endpgfgraphicnamed,   \quad
   \left\llbracket%
\beginpgfgraphicnamed{diagrams//generator_spider-nonum}
\InputIfFileExists{diagrams//generator_spider-nonum.tikz}{}{\input{./figures/diagrams//generator_spider-nonum.tikz}}
\endpgfgraphicnamed\right\rrbracket_{XW}=  %
\beginpgfgraphicnamed{diagrams//spiderwhite}
\InputIfFileExists{diagrams//spiderwhite.tikz}{}{\input{./figures/diagrams//spiderwhite.tikz}}
\endpgfgraphicnamed,   \quad 
     \left\llbracket%
\beginpgfgraphicnamed{diagrams//alphagate}
\begin{tikzpicture}
	\begin{pgfonlayer}{nodelayer}
		\node [style=none] (0) at (0, -0.5) {};
		\node [style=none] (1) at (0, 0.5) {};
		\node [style=gn] (2) at (0, 0) {$\alpha$};
	\end{pgfonlayer}
	\begin{pgfonlayer}{edgelayer}
		\draw (1.center) to (0.center);
	\end{pgfonlayer}
\end{tikzpicture}}
\endpgfgraphicnamed\right\rrbracket_{XW}=  %
\beginpgfgraphicnamed{diagrams//alphagatewhite}
\begin{tikzpicture}
	\begin{pgfonlayer}{nodelayer}
		\node [style=wn] (0) at (0.25, 0) {};
		\node [style=none] (1) at (0.25, -0.5) {};
		\node [style=none] (2) at (0.25, 0.5) {};
		\node [style=none] (3) at (0.5, 0) {$e^{i\alpha}$};
	\end{pgfonlayer}
	\begin{pgfonlayer}{edgelayer}
		\draw (2.center) to (1.center);
	\end{pgfonlayer}
\end{tikzpicture}}
\endpgfgraphicnamed,   \quad 
       \left\llbracket%
\beginpgfgraphicnamed{diagrams//lambdabox}
\begin{tikzpicture}
	\begin{pgfonlayer}{nodelayer}
		\node [style=H box] (0) at (0, 0) {$\lambda$};
		\node [style=none] (1) at (0, -0.5) {};
		\node [style=none] (2) at (0, 0.5) {};
	\end{pgfonlayer}
	\begin{pgfonlayer}{edgelayer}
		\draw (2.center) to (0);
		\draw (1.center) to (0);
	\end{pgfonlayer}
\end{tikzpicture}}
\endpgfgraphicnamed\right\rrbracket_{XW}=  %
\beginpgfgraphicnamed{diagrams//lambdagatewhiteld}
\begin{tikzpicture}
	\begin{pgfonlayer}{nodelayer}
		\node [style=none] (0) at (0, -0.5) {};
		\node [style=none] (1) at (0.25, 0) {$\lambda$};
		\node [style=wn] (2) at (0, 0) {};
		\node [style=none] (3) at (0, 0.5) {};
	\end{pgfonlayer}
	\begin{pgfonlayer}{edgelayer}
		\draw (3.center) to (0.center);
	\end{pgfonlayer}
\end{tikzpicture}}
\endpgfgraphicnamed,   
         \]

 \[
  \left\llbracket%
\beginpgfgraphicnamed{diagrams//HadaDecomSingleslt}
\begin{tikzpicture}
	\begin{pgfonlayer}{nodelayer}
		\node [style=H box] (0) at (-0.75, 0) {$H$};
		\node [style=none] (1) at (-0.75, -0.5) {};
		\node [style=none] (2) at (-0.75, 0.5) {};
	\end{pgfonlayer}
	\begin{pgfonlayer}{edgelayer}
		\draw (2.center) to (0);
		\draw (1.center) to (0);
	\end{pgfonlayer}
\end{tikzpicture}}
\endpgfgraphicnamed\right\rrbracket_{XW}=  %
\beginpgfgraphicnamed{diagrams//Hadamardwhite}
\InputIfFileExists{diagrams//Hadamardwhite.tikz}{}{\input{./figures/diagrams//Hadamardwhite.tikz}}
\endpgfgraphicnamed,   \quad
    \left\llbracket%
\beginpgfgraphicnamed{diagrams//triangle}
\begin{tikzpicture}
	\begin{pgfonlayer}{nodelayer}
		\node [style=none] (0) at (0, 0.5) {};
		\node [style=triangle] (1) at (0, 0) {};
		\node [style=none] (2) at (0, -0.5) {};
	\end{pgfonlayer}
	\begin{pgfonlayer}{edgelayer}
		\draw (0.center) to (2.center);
	\end{pgfonlayer}
\end{tikzpicture}}
\endpgfgraphicnamed\right\rrbracket_{XW}=  %
\beginpgfgraphicnamed{diagrams//trianglewhite}
\InputIfFileExists{diagrams//trianglewhite.tikz}{}{\input{./figures/diagrams//trianglewhite.tikz}}
\endpgfgraphicnamed, \]
    
    \[ \llbracket D_1\otimes D_2  \rrbracket_{XW} =  \llbracket D_1  \rrbracket_{XW} \otimes  \llbracket  D_2  \rrbracket_{XW}, \quad 
 \llbracket D_1\circ D_2  \rrbracket_{XW} =  \llbracket D_1  \rrbracket_{XW} \circ  \llbracket  D_2  \rrbracket_{XW},
 \]
where $ \alpha \in [0,~2\pi), ~ \lambda  \geq 0$.

\begin{lemma}\label{xtowpreservesemantics}
Suppose $D$ is an arbitrary diagram in ZX-calculus. Then  $\llbracket \llbracket D \rrbracket_{XW}\rrbracket = \llbracket D \rrbracket$.
\end{lemma}
The proof is easy.

Next we define the interpretation $\llbracket \cdot \rrbracket_{WX}$   from ZW-calculus to  ZX-calculus  as follows:

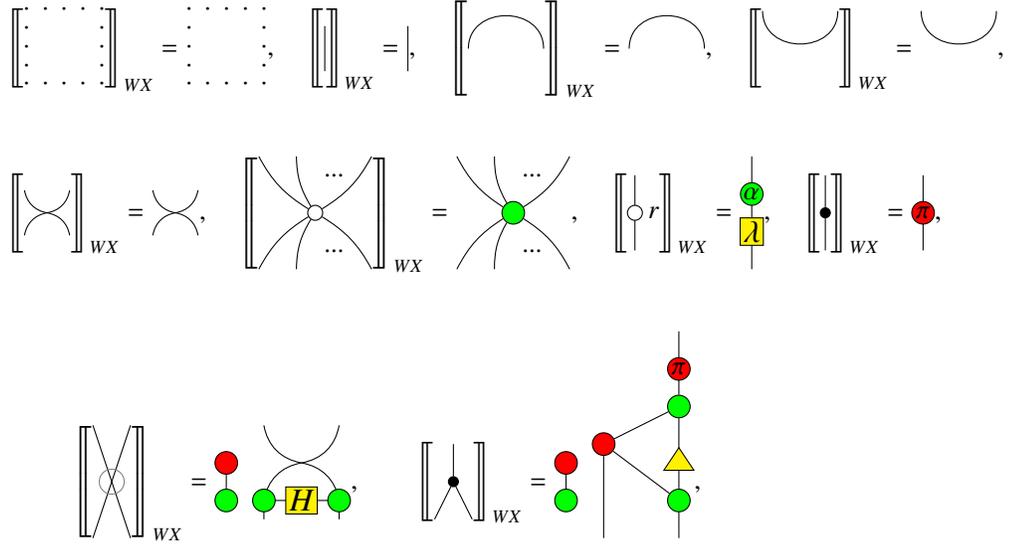
\[
 \left\llbracket%
\beginpgfgraphicnamed{diagrams//emptysquare}
\InputIfFileExists{diagrams//emptysquare.tikz}{}{\input{./figures/diagrams//emptysquare.tikz}}
\endpgfgraphicnamed\right\rrbracket_{WX}=  %
\beginpgfgraphicnamed{diagrams//emptysquare}
\InputIfFileExists{diagrams//emptysquare.tikz}{}{\input{./figures/diagrams//emptysquare.tikz}}
\endpgfgraphicnamed,  \quad
  \left\llbracket%
\beginpgfgraphicnamed{diagrams//Id}
\begin{tikzpicture}
	\begin{pgfonlayer}{nodelayer}
		\node [style=none] (1) at (0.5, 0.3) {};
		\node [style=none] (2) at (0.5, -0.3) {};
		\node [style=none] (3) at (0.5, -0.5) {};
		\node [style=none] (4) at (0.5, 0.5) {};
	\end{pgfonlayer}
	\begin{pgfonlayer}{edgelayer}
		\draw (1.center) to (2.center);
	\end{pgfonlayer}
\end{tikzpicture}}
\endpgfgraphicnamed\right\rrbracket_{WX}=  %
\beginpgfgraphicnamed{diagrams//Id}
\begin{tikzpicture}
	\begin{pgfonlayer}{nodelayer}
		\node [style=none] (1) at (0.5, 0.3) {};
		\node [style=none] (2) at (0.5, -0.3) {};
		\node [style=none] (3) at (0.5, -0.5) {};
		\node [style=none] (4) at (0.5, 0.5) {};
	\end{pgfonlayer}
	\begin{pgfonlayer}{edgelayer}
		\draw (1.center) to (2.center);
	\end{pgfonlayer}
\end{tikzpicture}}
\endpgfgraphicnamed,   \quad
 \left\llbracket%
\beginpgfgraphicnamed{diagrams//cap}
\begin{tikzpicture}
	\begin{pgfonlayer}{nodelayer}
		\node [style=none] (0) at (0, -0) {};
		\node [style=none] (1) at (1, -0) {};
	\end{pgfonlayer}
	\begin{pgfonlayer}{edgelayer}
		\draw [bend left=90, looseness=1.50] (0.center) to (1.center);
	\end{pgfonlayer}
\end{tikzpicture}}
\endpgfgraphicnamed\right\rrbracket_{WX}=  %
\beginpgfgraphicnamed{diagrams//cap}
\begin{tikzpicture}
	\begin{pgfonlayer}{nodelayer}
		\node [style=none] (0) at (0, -0) {};
		\node [style=none] (1) at (1, -0) {};
	\end{pgfonlayer}
	\begin{pgfonlayer}{edgelayer}
		\draw [bend left=90, looseness=1.50] (0.center) to (1.center);
	\end{pgfonlayer}
\end{tikzpicture}}
\endpgfgraphicnamed,   \quad
  \left\llbracket%
\beginpgfgraphicnamed{diagrams//cup}
\begin{tikzpicture}
	\begin{pgfonlayer}{nodelayer}
		\node [style=none] (0) at (0, 0.5) {};
		\node [style=none] (1) at (1, 0.5) {};
	\end{pgfonlayer}
	\begin{pgfonlayer}{edgelayer}
		\draw [bend right=90, looseness=1.50] (0.center) to (1.center);
	\end{pgfonlayer}
\end{tikzpicture}}
\endpgfgraphicnamed\right\rrbracket_{WX}=  %
\beginpgfgraphicnamed{diagrams//cup}
\begin{tikzpicture}
	\begin{pgfonlayer}{nodelayer}
		\node [style=none] (0) at (0, 0.5) {};
		\node [style=none] (1) at (1, 0.5) {};
	\end{pgfonlayer}
	\begin{pgfonlayer}{edgelayer}
		\draw [bend right=90, looseness=1.50] (0.center) to (1.center);
	\end{pgfonlayer}
\end{tikzpicture}}
\endpgfgraphicnamed,  
  \]
  
  \[
   \left\llbracket%
\beginpgfgraphicnamed{diagrams//swap}
\InputIfFileExists{diagrams//swap.tikz}{}{\input{./figures/diagrams//swap.tikz}}
\endpgfgraphicnamed\right\rrbracket_{WX}=  %
\beginpgfgraphicnamed{diagrams//swap}
\InputIfFileExists{diagrams//swap.tikz}{}{\input{./figures/diagrams//swap.tikz}}
\endpgfgraphicnamed,   \quad
   \left\llbracket%
\beginpgfgraphicnamed{diagrams//spiderwhite}
\InputIfFileExists{diagrams//spiderwhite.tikz}{}{\input{./figures/diagrams//spiderwhite.tikz}}
\endpgfgraphicnamed\right\rrbracket_{WX}=  %
\beginpgfgraphicnamed{diagrams//generator_spider-nonum}
\InputIfFileExists{diagrams//generator_spider-nonum.tikz}{}{\input{./figures/diagrams//generator_spider-nonum.tikz}}
\endpgfgraphicnamed,   \quad 
     \left\llbracket%
\beginpgfgraphicnamed{diagrams//rgatewhite}
\begin{tikzpicture}
	\begin{pgfonlayer}{nodelayer}
		\node [style=wn] (0) at (0, 0) {};
		\node [style=none] (1) at (0, -0.5) {};
		\node [style=none] (2) at (0.25, 0) {$r$};
		\node [style=none] (3) at (0, 0.5) {};
	\end{pgfonlayer}
	\begin{pgfonlayer}{edgelayer}
		\draw (3.center) to (1.center);
	\end{pgfonlayer}
\end{tikzpicture}}
\endpgfgraphicnamed\right\rrbracket_{WX}=  %
\beginpgfgraphicnamed{diagrams//alphalambdagate}
\begin{tikzpicture}
	\begin{pgfonlayer}{nodelayer}
		\node [style=H box] (0) at (0, -0.25) {$\lambda$};
		\node [style=none] (1) at (0, 0.75) {};
		\node [style=gn] (2) at (0, 0.25) {$\alpha$};
		\node [style=none] (3) at (0, -0.75) {};
	\end{pgfonlayer}
	\begin{pgfonlayer}{edgelayer}
		\draw (1.center) to (3.center);
	\end{pgfonlayer}
\end{tikzpicture}}
\endpgfgraphicnamed,   \quad 
       \left\llbracket%
\beginpgfgraphicnamed{diagrams//piblack}
\begin{tikzpicture}
	\begin{pgfonlayer}{nodelayer}
		\node [style=bn] (0) at (0, 0) {};
		\node [style=none] (1) at (0, -0.5) {};
		\node [style=none] (2) at (0, 0.5) {};
	\end{pgfonlayer}
	\begin{pgfonlayer}{edgelayer}
		\draw (2.center) to (1.center);
	\end{pgfonlayer}
\end{tikzpicture}}
\endpgfgraphicnamed\right\rrbracket_{WX}=  %
\beginpgfgraphicnamed{diagrams//pired}
\begin{tikzpicture}
	\begin{pgfonlayer}{nodelayer}
		\node [style=none] (0) at (0, -0.5) {};
		\node [style=none] (1) at (0, 0.5) {};
		\node [style=rn] (2) at (0, 0) {$\pi$};
	\end{pgfonlayer}
	\begin{pgfonlayer}{edgelayer}
		\draw (1.center) to (0.center);
	\end{pgfonlayer}
\end{tikzpicture}}
\endpgfgraphicnamed,   
         \]

 \[
  \left\llbracket%
\beginpgfgraphicnamed{diagrams//corsszw}
\InputIfFileExists{diagrams//corsszw.tikz}{}{\input{./figures/diagrams//corsszw.tikz}}
\endpgfgraphicnamed\right\rrbracket_{WX}=  %
\beginpgfgraphicnamed{diagrams//crossxz}
\InputIfFileExists{diagrams//crossxz.tikz}{}{\input{./figures/diagrams//crossxz.tikz}}
\endpgfgraphicnamed,   \quad \quad
    \left\llbracket%
\beginpgfgraphicnamed{diagrams//wblack}
\begin{tikzpicture}
	\begin{pgfonlayer}{nodelayer}
		\node [style=bn] (0) at (0, 0) {};
		\node [style=none] (1) at (-0.25, -0.5) {};
		\node [style=none] (2) at (0.25, -0.5) {};
		\node [style=none] (3) at (0, 0.5) {};
	\end{pgfonlayer}
	\begin{pgfonlayer}{edgelayer}
		\draw (3.center) to (0);
		\draw (0) to (1.center);
		\draw (0) to (2.center);
	\end{pgfonlayer}
\end{tikzpicture}}
\endpgfgraphicnamed\right\rrbracket_{WX}=  %
\beginpgfgraphicnamed{diagrams//winzx}
\InputIfFileExists{diagrams//winzx.tikz}{}{\input{./figures/diagrams//winzx.tikz}}
\endpgfgraphicnamed, \]

    \[ \llbracket D_1\otimes D_2  \rrbracket_{WX} =  \llbracket D_1  \rrbracket_{WX} \otimes  \llbracket  D_2  \rrbracket_{WX}, \quad 
 \llbracket D_1\circ D_2  \rrbracket_{WX} =  \llbracket D_1  \rrbracket_{WX} \circ  \llbracket  D_2  \rrbracket_{WX}.
 \]
where $r=\lambda e^{i\alpha},~   \alpha \in [0,~2\pi), ~ \lambda  \geq 0$.

\begin{lemma}\label{wtoxpreservesemantics}
Suppose $D$ is an arbitrary diagram in ZW-calculus. Then  $\llbracket \llbracket D \rrbracket_{WX}\rrbracket = \llbracket D \rrbracket$.
\end{lemma}
The proof is easy.

\begin{lemma}\label{interpretationreversible}
Suppose $D$ is an arbitrary diagram in ZX-calculus. Then  $ZX\vdash \llbracket \llbracket D \rrbracket_{XW}\rrbracket_{WX} =D$.
\end{lemma}

\begin{proof}
By the construction of  $\llbracket \cdot  \rrbracket_{XW}$  and $\llbracket \cdot  \rrbracket_{WX}$, we only need to prove for $D$ as a generator of ZX-calculus.  The first six generators in ZX-calculus are the same as the  first six generators in ZW-calculus, so we just check for the last four generators in ZX-calculus.

Since 
\[
 \left\llbracket%
\beginpgfgraphicnamed{diagrams//alphagate}
\begin{tikzpicture}
	\begin{pgfonlayer}{nodelayer}
		\node [style=none] (0) at (0, -0.5) {};
		\node [style=none] (1) at (0, 0.5) {};
		\node [style=gn] (2) at (0, 0) {$\alpha$};
	\end{pgfonlayer}
	\begin{pgfonlayer}{edgelayer}
		\draw (1.center) to (0.center);
	\end{pgfonlayer}
\end{tikzpicture}}
\endpgfgraphicnamed\right\rrbracket_{XW}=  %
\beginpgfgraphicnamed{diagrams//alphagatewhite}
\begin{tikzpicture}
	\begin{pgfonlayer}{nodelayer}
		\node [style=wn] (0) at (0.25, 0) {};
		\node [style=none] (1) at (0.25, -0.5) {};
		\node [style=none] (2) at (0.25, 0.5) {};
		\node [style=none] (3) at (0.5, 0) {$e^{i\alpha}$};
	\end{pgfonlayer}
	\begin{pgfonlayer}{edgelayer}
		\draw (2.center) to (1.center);
	\end{pgfonlayer}
\end{tikzpicture}}
\endpgfgraphicnamed,
\]
we have 
\[
 \left\llbracket \left\llbracket%
\beginpgfgraphicnamed{diagrams//alphagate}
\begin{tikzpicture}
	\begin{pgfonlayer}{nodelayer}
		\node [style=none] (0) at (0, -0.5) {};
		\node [style=none] (1) at (0, 0.5) {};
		\node [style=gn] (2) at (0, 0) {$\alpha$};
	\end{pgfonlayer}
	\begin{pgfonlayer}{edgelayer}
		\draw (1.center) to (0.center);
	\end{pgfonlayer}
\end{tikzpicture}}
\endpgfgraphicnamed\right\rrbracket_{XW}\right\rrbracket_{WX}= \left\llbracket%
\beginpgfgraphicnamed{diagrams//alphagatewhite}
\begin{tikzpicture}
	\begin{pgfonlayer}{nodelayer}
		\node [style=wn] (0) at (0.25, 0) {};
		\node [style=none] (1) at (0.25, -0.5) {};
		\node [style=none] (2) at (0.25, 0.5) {};
		\node [style=none] (3) at (0.5, 0) {$e^{i\alpha}$};
	\end{pgfonlayer}
	\begin{pgfonlayer}{edgelayer}
		\draw (2.center) to (1.center);
	\end{pgfonlayer}
\end{tikzpicture}}
\endpgfgraphicnamed\right\rrbracket_{WX}=%
\beginpgfgraphicnamed{diagrams//alphagate}
\begin{tikzpicture}
	\begin{pgfonlayer}{nodelayer}
		\node [style=none] (0) at (0, -0.5) {};
		\node [style=none] (1) at (0, 0.5) {};
		\node [style=gn] (2) at (0, 0) {$\alpha$};
	\end{pgfonlayer}
	\begin{pgfonlayer}{edgelayer}
		\draw (1.center) to (0.center);
	\end{pgfonlayer}
\end{tikzpicture}}
\endpgfgraphicnamed,
\]
by the definition of  $\llbracket \cdot  \rrbracket_{WX}$ and the ZX rule (L3).
Similarly, we can easily check that 
\[
 \left\llbracket \left\llbracket%
\beginpgfgraphicnamed{diagrams//lambdabox}
\begin{tikzpicture}
	\begin{pgfonlayer}{nodelayer}
		\node [style=H box] (0) at (0, 0) {$\lambda$};
		\node [style=none] (1) at (0, -0.5) {};
		\node [style=none] (2) at (0, 0.5) {};
	\end{pgfonlayer}
	\begin{pgfonlayer}{edgelayer}
		\draw (2.center) to (0);
		\draw (1.center) to (0);
	\end{pgfonlayer}
\end{tikzpicture}}
\endpgfgraphicnamed\right\rrbracket_{XW}\right\rrbracket_{WX}=%
\beginpgfgraphicnamed{diagrams//lambdabox}
\begin{tikzpicture}
	\begin{pgfonlayer}{nodelayer}
		\node [style=H box] (0) at (0, 0) {$\lambda$};
		\node [style=none] (1) at (0, -0.5) {};
		\node [style=none] (2) at (0, 0.5) {};
	\end{pgfonlayer}
	\begin{pgfonlayer}{edgelayer}
		\draw (2.center) to (0);
		\draw (1.center) to (0);
	\end{pgfonlayer}
\end{tikzpicture}}
\endpgfgraphicnamed,\quad 
  \left\llbracket \left\llbracket%
\beginpgfgraphicnamed{diagrams//HadaDecomSingleslt}
\begin{tikzpicture}
	\begin{pgfonlayer}{nodelayer}
		\node [style=H box] (0) at (-0.75, 0) {$H$};
		\node [style=none] (1) at (-0.75, -0.5) {};
		\node [style=none] (2) at (-0.75, 0.5) {};
	\end{pgfonlayer}
	\begin{pgfonlayer}{edgelayer}
		\draw (2.center) to (0);
		\draw (1.center) to (0);
	\end{pgfonlayer}
\end{tikzpicture}}
\endpgfgraphicnamed\right\rrbracket_{XW}\right\rrbracket_{WX}=%
\beginpgfgraphicnamed{diagrams//HadaDecomSingleslt}
\begin{tikzpicture}
	\begin{pgfonlayer}{nodelayer}
		\node [style=H box] (0) at (-0.75, 0) {$H$};
		\node [style=none] (1) at (-0.75, -0.5) {};
		\node [style=none] (2) at (-0.75, 0.5) {};
	\end{pgfonlayer}
	\begin{pgfonlayer}{edgelayer}
		\draw (2.center) to (0);
		\draw (1.center) to (0);
	\end{pgfonlayer}
\end{tikzpicture}}
\endpgfgraphicnamed,\quad 
   \left\llbracket \left\llbracket%
\beginpgfgraphicnamed{diagrams//triangle}
\begin{tikzpicture}
	\begin{pgfonlayer}{nodelayer}
		\node [style=none] (0) at (0, 0.5) {};
		\node [style=triangle] (1) at (0, 0) {};
		\node [style=none] (2) at (0, -0.5) {};
	\end{pgfonlayer}
	\begin{pgfonlayer}{edgelayer}
		\draw (0.center) to (2.center);
	\end{pgfonlayer}
\end{tikzpicture}}
\endpgfgraphicnamed\right\rrbracket_{XW}\right\rrbracket_{WX}=%
\beginpgfgraphicnamed{diagrams//triangle}
\begin{tikzpicture}
	\begin{pgfonlayer}{nodelayer}
		\node [style=none] (0) at (0, 0.5) {};
		\node [style=triangle] (1) at (0, 0) {};
		\node [style=none] (2) at (0, -0.5) {};
	\end{pgfonlayer}
	\begin{pgfonlayer}{edgelayer}
		\draw (0.center) to (2.center);
	\end{pgfonlayer}
\end{tikzpicture}}
\endpgfgraphicnamed.
\]

\end{proof}

\section{Completeness}
\begin{proposition}\label{zwrulesholdinzx}
If  $ZW\vdash D_1=D_2$, then  $ZX\vdash \llbracket D_1 \rrbracket_{WX} =\llbracket D_2 \rrbracket_{WX}$.

\end{proposition}

\begin{proof}
Here we need only to prove that $ZX \vdash \left\llbracket D_1\right\rrbracket_{WX} = \left\llbracket D_2\right\rrbracket_{WX}$ where  $D_1=D_2$ is a rewriting rule of ZW-calculus. The whole of Appendix is devoted to prove this proposition.
\end{proof}

\begin{theorem}\label{maintheorem}
The ZX-calculus is complete for universal pure qubit quantum mechanics:
If $\llbracket D_1 \rrbracket =\llbracket D_2 \rrbracket$, then $ZX\vdash D_1=D_2$,

\end{theorem}

\begin{proof}
Suppose $D_1,  D_2 \in ZX$ and  $\llbracket D_1 \rrbracket =\llbracket D_2 \rrbracket$. Then by lemma \ref{xtowpreservesemantics},  $\llbracket \llbracket D_1 \rrbracket_{XW}\rrbracket = \llbracket D_1 \rrbracket= \llbracket D_2 \rrbracket=\llbracket \llbracket D_2 \rrbracket_{XW}\rrbracket $.  Thus by the completeness of ZW-calculus \cite{amar},  $ZW\vdash \llbracket D_2 \rrbracket_{XW}=  \llbracket D_2 \rrbracket_{XW}$.  Now by proposition \ref{zwrulesholdinzx},  $ZX\vdash \llbracket \llbracket D_1 \rrbracket_{XW}\rrbracket_{WX} =\llbracket \llbracket D_2 \rrbracket_{XW}\rrbracket_{WX}$.
Finally, by lemma \ref{interpretationreversible},  $D_1=D_2$.
\end{proof}


\section{Conclusion and further work}

In this paper, we  show that  the ZX-calculus is complete for the universal pure qubit QM, with the aid of completeness of ZW-calculus for the whole qubit QM. 

There are several questions for the next step. Firstly, can we derive the completeness of ZX-calculus for Clifford+T QM from the universal completeness? Secondly, can we obtain the completeness of ZX-calculus for stabilizer QM from the universal completeness?  Thirdly, 
can we generalise the completeness result to qudit ZX-calculus for arbitrary dimension $d$? Furthermore, can we have a proof of completeness that is independent of the ZW-calculus?

It is also interesting to incorporate the rules of the universally complete  ZX-calculus in the automated graph rewriting system Quantomatic  \cite{Quanto}.

\section*{Acknowledgement}
 The authors would like to thank  Bob Coecke and Amar Hadzihasanovic for the fruitful discussions and invaluable comments.

\section*{Appendix}
\begin{lemma}~\newline
	\label{lem:6b}
	$

	$
\end{lemma}

\begin{proof}
	Proof in \cite{bpw}. The rules used are $S1$, $S2$, $S3$, $H$, $H2$, $H3$, $B1$, $B2$, $EU$, $\ref{lem:inv}$.
\end{proof}

\begin{lemma}~\newline
	\label{lem:hopf}
	$
	%
	$
\end{lemma}

\begin{proof}
	Proof in \cite{bpw}. The rules used are $S1$, $S2$, $S3$, $H2$, $H3$, $H$, $B1$, $B2$.
\end{proof}

\begin{lemma}~\newline
	\label{lem:com}
	$
	%
	$
\end{lemma}

\begin{proof}
	Proof in \cite{bpw}. The rules used are \ref{lem:hopf}, $B2$, $S2$, $H2$, $S1$, $H$, $B1$.
\end{proof}

\begin{lemma}~\newline
	\label{lem:1}
	$
	%
	$
\end{lemma}

\begin{proof}~\newline
	Proof in \cite{Emmanuel} lemma $33$. The rules used are \ref{lem:hopf}, \ref{lem:inv}, $B1$, $TR2$.
\end{proof}

\begin{lemma}~\newline
	\label{lem:3}
	$
	%
	$
\end{lemma}

\begin{proof}~\newline
	Proof in \cite{Emmanuel} lemma $19$. The rules used are \ref{lem:hopf}, \ref{lem:inv}, $B1$, $TR2$.
\end{proof}

\begin{lemma}~\newline
	\label{lem:4}
	$
	%
	$
\end{lemma}

\begin{proof}~\newline
	Proof in \cite{Emmanuel} lemma $26$. The rules used are $B2, S1, TR7$.
\end{proof}

\begin{lemma}~\newline
	\label{lem:5}
	$
	%
	$
\end{lemma}

\begin{proof}~\newline
	Straightforward application of rules $\ref{lem:com}$ and $TR8$.
\end{proof}

\begin{lemma}~\newline
	\label{lem:6}
	$
	%
	$
\end{lemma}

\begin{proof}~\newline
	Straightforward application of rules $H$ and $EU$.
\end{proof}

\begin{lemma}~\newline
	\label{lem:7}
	$
	%
	~\right\rrbracket_{WX}
	$
\end{proposition}

\begin{proof}
	Similar to proposition \ref{prop:natnx}.
\end{proof}
\begin{proposition}\label{prop:reix1}(ZW rule $rei^x_1$)~\newline \\
	$
	ZX\vdash
	\left\llbracket~
	%
	~\right\rrbracket_{WX}
	$
\end{proposition}

\begin{proof}
	Similar proof in \cite{Emmanuel}, proposition 7 part $1b$. The rules used are $\ref{lem:inv}$, $S1$, $H2$, $H$, $S2$, $TR7$, $\ref{lem:hopf}$, $TR2$, $\ref{lem:3}$.
\end{proof}
\begin{proposition}(ZW rule $un^{co}_{w,R}$)~\newline \\
	$
	ZX\vdash
	\left\llbracket~
	%
	~\right\rrbracket_{WX}
	$
\end{proposition}

\begin{proof}~\newline
	Proof in \cite{Emmanuel}, proposition 7 part $1b$. The rules used are $\ref{lem:inv}$, $S1$, $H2$, $H$, $S2$, $TR7$, $\ref{lem:hopf}$, $TR2$, $\ref{lem:3}$.
\end{proof}
\begin{proposition}(ZW rule $nat^w_w$)~\newline\\
	$
	ZX\vdash
	\left\llbracket~
	\input{Def_5.4/natww/natww_LHS.tikz}
	~\right\rrbracket_{WX}
	=
	\left\llbracket~
	\input{Def_5.4/natww/natww_RHS.tikz}
	~\right\rrbracket_{WX}
	$
\end{proposition}

\begin{proof}
	Claim:\\
	\input{Def_5.4/natww/natww_claim.tikz}\\
	Proof of claim:\\
	$`\implies'$:\\
	\input{Def_5.4/natww/natww_claimproof.tikz}\\
	$`\impliedby'$:\\
	\input{Def_5.4/natww/natww_claimproof2.tikz}\\
	Hence we only have to prove\\
		$
	ZX\vdash
	\left\llbracket~
	\input{Def_5.4/natww/natww_simon_LHS.tikz}
	~\right\rrbracket_{WX}
	=
	\left\llbracket~
	\input{Def_5.4/natww/natww_simon_RHS.tikz}
	~\right\rrbracket_{WX}
	$
	\\
	which is proven in \cite{Emmanuel}, proposition 7 part $1a$. The rules used are $S1$, $B2$, $\ref{lem:4}$, $\ref{lem:5}$, $TR8$.
\end{proof}
\begin{proposition}(ZW rule $nat^w_x$)~\newline\\
	$
	ZX\vdash
	\left\llbracket~
	\begin{tikzpicture}
	\begin{pgfonlayer}{nodelayer}
		\node [style=none] (0) at (-0.5, -0) {};
		\node [style=dn] (1) at (-0.5, -0.375) {};
		\node [style=bn] (2) at (0, 0.75) {};
		\node [style=none] (3) at (0.7500001, 0.25) {};
		\node [style=none] (4) at (-0.7500001, -0.7499999) {};
		\node [style=bn] (5) at (0, 0.5) {};
		\node [style=none] (6) at (0.7500001, 1) {};
		\node [style=none] (7) at (0.5000002, -0) {};
		\node [style=none] (8) at (-0.5000001, -1) {};
		\node [style=none] (9) at (0, 1) {};
		\node [style=dn] (10) at (0.5000002, -0.125) {};
		\node [style=none] (11) at (0.5000001, -1) {};
		\node [style=none] (12) at (-0.7499999, -1) {};
	\end{pgfonlayer}
	\begin{pgfonlayer}{edgelayer}
		\draw (8.center) to (0.center);
		\draw (11.center) to (7.center);
		\draw [in=-165, out=90, looseness=0.75] (0.center) to (5);
		\draw [in=90, out=-15, looseness=0.75] (5) to (7.center);
		\draw (5) to (2);
		\draw (2) to (9.center);
		\draw [in=-90, out=90, looseness=1.00] (4.center) to (3.center);
		\draw (3.center) to (6.center);
		\draw (4.center) to (12.center);
	\end{pgfonlayer}
\end{tikzpicture}
	~\right\rrbracket_{WX}
	=
	\left\llbracket~
	\begin{tikzpicture}
	\begin{pgfonlayer}{nodelayer}
		\node [style=none] (0) at (-0.5000001, -1) {};
		\node [style=none] (1) at (0, 1) {};
		\node [style=none] (2) at (-0.7499999, -0.7500001) {};
		\node [style=bn] (3) at (0, -0) {};
		\node [style=bn] (4) at (0, -0.2499999) {};
		\node [style=none] (5) at (0.7500002, 1) {};
		\node [style=dn] (6) at (0, 0.5) {};
		\node [style=none] (7) at (0.5000001, -1) {};
		\node [style=none] (8) at (-0.7499999, -0) {};
		\node [style=none] (9) at (-0.7499999, -1) {};
	\end{pgfonlayer}
	\begin{pgfonlayer}{edgelayer}
		\draw [in=90, out=-165, looseness=0.75] (4) to (0.center);
		\draw [in=90, out=-15, looseness=0.75] (4) to (7.center);
		\draw (4) to (3);
		\draw (3) to (1.center);
		\draw [in=90, out=-90, looseness=1.25] (5.center) to (8.center);
		\draw (8.center) to (2.center);
		\draw (2.center) to (9.center);
	\end{pgfonlayer}
\end{tikzpicture}
	~\right\rrbracket_{WX}
	$
\end{proposition}

\begin{proof}
	Proof in \cite{Emmanuel}, proposition 7 part $7a$. The rules used are $H$, $\ref{lem:1}$, $S1$, $B2$.
\end{proof}
\begin{proposition}(ZW rule $com^{co}_w$)~\newline \\
	$
	ZX\vdash
	\left\llbracket~
	\begin{tikzpicture}
	\begin{pgfonlayer}{nodelayer}
		\node [style=none] (0) at (-0.5, -0) {};
		\node [style=none] (1) at (-0.5, -1) {};
		\node [style=bn] (2) at (0, 0.75) {};
		\node [style=none] (3) at (0, 1) {};
		\node [style=bn] (4) at (0, 0.5) {};
		\node [style=dn] (5) at (0, -0.5) {};
		\node [style=none] (6) at (0.5, -1) {};
		\node [style=none] (7) at (0.5, -0) {};
	\end{pgfonlayer}
	\begin{pgfonlayer}{edgelayer}
		\draw [in=-165, out=90, looseness=0.75] (0.center) to (4);
		\draw [in=90, out=-15, looseness=0.75] (4) to (7.center);
		\draw (4) to (2);
		\draw (2) to (3.center);
		\draw [in=-90, out=90, looseness=1.00] (1.center) to (7.center);
		\draw [in=90, out=-90, looseness=1.00] (0.center) to (6.center);
	\end{pgfonlayer}
\end{tikzpicture}
	~\right\rrbracket_{WX}
	=
	\left\llbracket~
	\begin{tikzpicture}
	\begin{pgfonlayer}{nodelayer}
		\node [style=none] (0) at (0, 1) {};
		\node [style=bn] (1) at (0, 0.75) {};
		\node [style=none] (2) at (0.5, -1) {};
		\node [style=bn] (3) at (0, 0.5) {};
		\node [style=none] (4) at (0.5, -0) {};
		\node [style=none] (5) at (-0.5, -1) {};
		\node [style=none] (6) at (-0.5, -0) {};
	\end{pgfonlayer}
	\begin{pgfonlayer}{edgelayer}
		\draw [in=-165, out=90, looseness=0.75] (6.center) to (3);
		\draw [in=90, out=-15, looseness=0.75] (3) to (4.center);
		\draw (3) to (1);
		\draw (1) to (0.center);
		\draw (5.center) to (6.center);
		\draw (4.center) to (2.center);
	\end{pgfonlayer}
\end{tikzpicture}
	~\right\rrbracket_{WX}
	$
\end{proposition}

\begin{proof}~\newline
	\begin{tikzpicture}
	\begin{pgfonlayer}{nodelayer}
		\node [style=none] (0) at (4, 12.75) {};
		\node [style=rn] (1) at (4, 12.25) {$\pi$};
		\node [style=rn] (2) at (4, 11.75) {$\pi$};
		\node [style=gn] (3) at (4, 11.25) {};
		\node [style=rn] (4) at (3, 10.25) {};
		\node [style=gn] (5) at (4, 10.25) {};
		\node [style=gn] (6) at (3, 8.75) {};
		\node [style=gn] (7) at (4, 8.75) {};
		\node [style=none] (8) at (3, 9.75) {};
		\node [style=none] (9) at (4, 9.75) {};
		\node [style={H box}] (10) at (3.5, 8.75) {$H$};
		\node [style=triangle] (11) at (4, 10.75) {};
		\node [style=none] (12) at (3, 8.25) {};
		\node [style=none] (13) at (4, 8.25) {};
		\node [style=gn] (14) at (4.5, 11) {};
		\node [style=gn] (15) at (4.5, 10) {};
		\node [style=rn] (16) at (4.5, 10.5) {};
		\node [style=rn] (17) at (4.5, 9.5) {};
		\node [style=gn] (18) at (6, 9.25) {};
		\node [style=rn] (19) at (6, 9.75) {};
		\node [style=none] (20) at (7, 11.5) {};
		\node [style=gn] (21) at (7.5, 11) {};
		\node [style=gn] (22) at (7, 9.75) {};
		\node [style=gn] (23) at (7.5, 10) {};
		\node [style={H box}] (24) at (6.5, 9.25) {$H$};
		\node [style=gn] (25) at (7, 9.25) {};
		\node [style=gn] (26) at (7, 10.75) {};
		\node [style=rn] (27) at (7.5, 10.5) {};
		\node [style=triangle] (28) at (7, 10.25) {};
		\node [style=rn] (29) at (7.5, 9.5) {};
		\node [style=none] (30) at (6, 8.25) {};
		\node [style=none] (31) at (7, 8.25) {};
		\node [style=none] (32) at (5.25, 10.25) {$\overset{\ref{lem:1}}{=}$};
		\node [style=rn] (33) at (11, 9.5) {};
		\node [style={H box}] (34) at (10, 9.25) {$H$};
		\node [style=none] (35) at (10.5, 8.25) {};
		\node [style=rn] (36) at (11, 10.5) {};
		\node [style=gn] (37) at (11, 11) {};
		\node [style=gn] (38) at (11, 10) {};
		\node [style=gn] (39) at (10.5, 11) {};
		\node [style=triangle] (40) at (10.5, 10.25) {};
		\node [style=gn] (41) at (10.5, 9.25) {};
		\node [style=none] (42) at (9.5, 8.25) {};
		\node [style=none] (43) at (10.5, 11.5) {};
		\node [style=gn] (44) at (10.5, 9.75) {};
		\node [style=gn] (45) at (9, 9.75) {};
		\node [style=gn] (46) at (9.5, 9.75) {};
		\node [style=rn] (47) at (9, 9.25) {};
		\node [style=rn] (48) at (9.5, 9.25) {};
		\node [style=none] (49) at (8.249999, 10.25) {$\overset{B2}{=}$};
		\node [style=gn] (50) at (11, 9) {};
		\node [style=rn] (51) at (11, 8.5) {};
		\node [style=none] (52) at (2.25, 7.25) {};
		\node [style=gn] (53) at (2.75, 4.75) {};
		\node [style=none] (54) at (1.25, 4) {};
		\node [style=gn] (55) at (2.75, 6.75) {};
		\node [style=rn] (56) at (2.75, 5.25) {};
		\node [style=gn] (57) at (2.75, 5.75) {};
		\node [style=triangle] (58) at (2.25, 6) {};
		\node [style=gn] (59) at (2.25, 5.5) {};
		\node [style=rn] (60) at (2.75, 6.25) {};
		\node [style=none] (61) at (2.25, 4) {};
		\node [style=rn] (62) at (2.75, 4.25) {};
		\node [style=gn] (63) at (2.25, 6.75) {};
		\node [style={H box}] (64) at (1.75, 5.25) {$H$};
		\node [style=rn] (65) at (1.25, 5.5) {};
		\node [style=gn] (66) at (0.7500001, 5.5) {};
		\node [style=rn] (67) at (0.7500001, 4.5) {};
		\node [style=none] (68) at (0, 5.75) {$\overset{S1}{=}$};
		\node [style=none] (69) at (5.25, 4) {};
		\node [style={H box}] (70) at (4.75, 5.25) {$H$};
		\node [style=triangle] (71) at (5.25, 5.75) {};
		\node [style=gn] (72) at (5.75, 6.5) {};
		\node [style=gn] (73) at (5.75, 5.5) {};
		\node [style=gn] (74) at (5.25, 5.25) {$\pi$};
		\node [style=rn] (75) at (4.25, 4.75) {};
		\node [style=none] (76) at (5.25, 7) {};
		\node [style=gn] (77) at (5.25, 6.5) {};
		\node [style=rn] (78) at (5.75, 5) {};
		\node [style=none] (79) at (3.5, 5.75) {$\overset{\ref{lem:6}}{=}$};
		\node [style=rn] (80) at (5.75, 6) {};
		\node [style=gn] (81) at (4.25, 5.25) {};
		\node [style=none] (82) at (4.25, 4) {};
		\node [style=gn] (83) at (8.000001, 5.25) {};
		\node [style=gn] (84) at (8.5, 6.5) {};
		\node [style=gn] (85) at (8.000001, 6.5) {};
		\node [style=rn] (86) at (8.5, 6) {};
		\node [style=triangle] (87) at (8.000001, 5.75) {};
		\node [style=none] (88) at (8.000001, 7) {};
		\node [style=rn] (89) at (7.25, 5.25) {};
		\node [style=none] (90) at (8.000001, 4) {};
		\node [style=none] (91) at (7.25, 4) {};
		\node [style=none] (92) at (6.5, 5.75) {$\overset{S1, TR5,}{=}$};
		\node [style=none] (93) at (1.5, 9.75) {};
		\node [style=bn] (94) at (0.7500001, 10.5) {};
		\node [style=none] (95) at (0, 8.25) {};
		\node [style=none] (96) at (0, 9.75) {};
		\node [style=bn] (97) at (0.7500001, 11) {};
		\node [style=none] (98) at (0.7500001, 11.25) {};
		\node [style=dn] (99) at (0.7500001, 9) {};
		\node [style=none] (100) at (1.5, 8.25) {};
		\node [style=none] (101) at (2.25, 10.25) {$\mapsto$};
		\node [style=none] (102) at (0.7500001, 1.25) {};
		\node [style=none] (103) at (2.25, -0) {};
		\node [style=none] (104) at (1.5, 3) {};
		\node [style=bn] (105) at (1.5, 2) {};
		\node [style=bn] (106) at (1.5, 2.5) {};
		\node [style=none] (107) at (2.25, 1.25) {};
		\node [style=none] (108) at (0.7500001, -0) {};
		\node [style=none] (109) at (0, 1.5) {$\mapsfrom$};
		\node [style=rn] (110) at (10, 5.25) {};
		\node [style=gn] (111) at (11.25, 6.5) {};
		\node [style=none] (112) at (10, 4) {};
		\node [style=none] (113) at (10.75, 4) {};
		\node [style=gn] (114) at (10.75, 6.5) {};
		\node [style=none] (115) at (10.75, 7) {};
		\node [style=gn] (116) at (10.75, 5.25) {};
		\node [style=rn] (117) at (11.25, 6) {};
		\node [style=triangle] (118) at (10.75, 5.75) {};
		\node [style=none] (119) at (9.25, 5.75) {$\overset{\ref{lem:com}, \ref{lem:4}}{=}$};
	\end{pgfonlayer}
	\begin{pgfonlayer}{edgelayer}
		\draw (0.center) to (1);
		\draw (1) to (2);
		\draw (2) to (3);
		\draw (3) to (11);
		\draw (11) to (5);
		\draw (5) to (9.center);
		\draw (4) to (8.center);
		\draw [in=90, out=-90, looseness=1.00] (8.center) to (7);
		\draw [in=-90, out=90, looseness=1.00] (6) to (9.center);
		\draw (6) to (12.center);
		\draw (7) to (13.center);
		\draw (6) to (10);
		\draw (10) to (7);
		\draw (4) to (5);
		\draw (14) to (16);
		\draw (15) to (17);
		\draw (26) to (28);
		\draw (28) to (22);
		\draw (18) to (24);
		\draw (24) to (25);
		\draw (19) to (22);
		\draw (21) to (27);
		\draw (23) to (29);
		\draw (18) to (19);
		\draw (25) to (22);
		\draw [in=-90, out=90, looseness=1.00] (30.center) to (25);
		\draw [in=-90, out=90, looseness=1.00] (31.center) to (18);
		\draw (39) to (40);
		\draw (40) to (44);
		\draw (34) to (41);
		\draw (37) to (36);
		\draw (38) to (33);
		\draw (41) to (44);
		\draw [in=-90, out=90, looseness=1.00] (42.center) to (41);
		\draw (45) to (47);
		\draw (45) to (48);
		\draw (47) to (46);
		\draw (48) to (46);
		\draw (4) to (3);
		\draw (19) to (26);
		\draw (45) to (39);
		\draw (46) to (44);
		\draw (48) to (34);
		\draw [in=90, out=-90, looseness=0.75] (47) to (35.center);
		\draw (50) to (51);
		\draw (63) to (58);
		\draw (64) to (59);
		\draw (55) to (60);
		\draw (57) to (56);
		\draw [in=-90, out=90, looseness=0.75] (54.center) to (59);
		\draw (53) to (62);
		\draw (64) to (65);
		\draw [bend left, looseness=1.00] (65) to (59);
		\draw (66) to (63);
		\draw (63) to (52.center);
		\draw (26) to (20.center);
		\draw (39) to (43.center);
		\draw (66) to (65);
		\draw (67) to (66);
		\draw [bend right, looseness=0.75] (67) to (59);
		\draw [in=90, out=-18, looseness=0.75] (67) to (61.center);
		\draw (58) to (59);
		\draw (77) to (71);
		\draw (70) to (74);
		\draw (72) to (80);
		\draw (73) to (78);
		\draw [in=-90, out=90, looseness=0.75] (82.center) to (74);
		\draw (81) to (77);
		\draw (77) to (76.center);
		\draw (75) to (81);
		\draw [bend right, looseness=0.75] (75) to (74);
		\draw [in=90, out=-18, looseness=0.75] (75) to (69.center);
		\draw (71) to (74);
		\draw (70) to (81);
		\draw (85) to (87);
		\draw (84) to (86);
		\draw (85) to (88.center);
		\draw (87) to (83);
		\draw (85) to (89);
		\draw (89) to (83);
		\draw [in=90, out=-90, looseness=1.00] (83) to (91.center);
		\draw [in=90, out=-90, looseness=1.00] (89) to (90.center);
		\draw [in=-165, out=90, looseness=0.75] (96.center) to (94);
		\draw [in=90, out=-15, looseness=0.75] (94) to (93.center);
		\draw (94) to (97);
		\draw (97) to (98.center);
		\draw [in=-90, out=90, looseness=1.00] (95.center) to (93.center);
		\draw [in=90, out=-90, looseness=1.00] (96.center) to (100.center);
		\draw [in=-165, out=90, looseness=0.75] (102.center) to (105);
		\draw [in=90, out=-15, looseness=0.75] (105) to (107.center);
		\draw (105) to (106);
		\draw (106) to (104.center);
		\draw (108.center) to (102.center);
		\draw (107.center) to (103.center);
		\draw (114) to (118);
		\draw (111) to (117);
		\draw (114) to (115.center);
		\draw (118) to (116);
		\draw (114) to (110);
		\draw (110) to (116);
		\draw [in=90, out=-90, looseness=1.00] (116) to (113.center);
		\draw [in=90, out=-90, looseness=1.00] (110) to (112.center);
	\end{pgfonlayer}
\end{tikzpicture}\\
\end{proof}
\begin{proposition}(ZW rule $nat^m_w$)~\newline \\
	$
	ZX\vdash
	\left\llbracket~
	\begin{tikzpicture}
	\begin{pgfonlayer}{nodelayer}
		\node [style=bn] (0) at (-0.5, -0.5) {};
		\node [style=bn] (1) at (-0.5, -0.7499999) {};
		\node [style=bn] (2) at (0.5000002, 0.5) {};
		\node [style=dn] (3) at (0, -0) {};
		\node [style=none] (4) at (0.5000002, 1) {};
		\node [style=none] (5) at (-0.5, 1) {};
		\node [style=bn] (6) at (0.5000002, -0.7499999) {};
		\node [style=bn] (7) at (-0.5, 0.75) {};
		\node [style=bn] (8) at (0.5000002, 0.75) {};
		\node [style=none] (9) at (0.5000002, -1) {};
		\node [style=bn] (10) at (-0.5, 0.5) {};
		\node [style=none] (11) at (-0.5, -1) {};
		\node [style=bn] (12) at (0.5000002, -0.5) {};
	\end{pgfonlayer}
	\begin{pgfonlayer}{edgelayer}
		\draw (10) to (12);
		\draw [bend right, looseness=1.00] (12) to (2);
		\draw (0) to (2);
		\draw [bend right, looseness=1.00] (10) to (0);
		\draw (0) to (1);
		\draw (12) to (6);
		\draw (2) to (8);
		\draw (10) to (7);
		\draw (7) to (5.center);
		\draw (8) to (4.center);
		\draw (1) to (11.center);
		\draw (6) to (9.center);
	\end{pgfonlayer}
\end{tikzpicture}
	~\right\rrbracket_{WX}
	=
	\left\llbracket~
	\begin{tikzpicture}
	\begin{pgfonlayer}{nodelayer}
		\node [style=bn] (0) at (0, 0.5) {};
		\node [style=none] (1) at (0.5000002, -0.7499999) {};
		\node [style=none] (2) at (0.5000002, 1) {};
		\node [style=bn] (3) at (0, -0.25) {};
		\node [style=none] (4) at (-0.5, 1) {};
		\node [style=none] (5) at (-0.5, -0.7499999) {};
		\node [style=bn] (6) at (0, -0) {};
		\node [style=bn] (7) at (0, 0.25) {};
	\end{pgfonlayer}
	\begin{pgfonlayer}{edgelayer}
		\draw [in=-90, out=150, looseness=0.75] (0) to (4.center);
		\draw [in=-90, out=30, looseness=0.75] (0) to (2.center);
		\draw (7) to (0);
		\draw (7) to (6);
		\draw (6) to (3);
		\draw [in=90, out=-150, looseness=0.75] (3) to (5.center);
		\draw [in=90, out=-30, looseness=0.75] (3) to (1.center);
	\end{pgfonlayer}
\end{tikzpicture}
	~\right\rrbracket_{WX}
	$
\end{proposition}

\begin{proof}
	Proof in \cite{Emmanuel}, proposition 7, part $5a$. It was proved in eight parts, (\textit{i}), (\textit{ii}), (\textit{iii}), (\textit{iv}), (\textit{v}), (\textit{vi}), (\textit{vii}), (\textit{viii}).
	
	The rules used in part (\textit{i}) are $B2$, $S1$, $H$, $TR5$, $\ref{lem:6b}$.
	
	The rules used in part (\textit{ii}) are $\ref{lem:4}$, $S1$, $B2$, $S1$, $TR11$, $\ref{lem:6b}$, $TR1$.
	
	The rules used in part (\textit{iii}) are $TR6$, $\ref{lem:4}$, $S1$, $TR8$, $\ref{lem:6b}$.
	
	The rules used in part (\textit{iv}) are $\ref{lem:4}$, $S1$, $B2$, $\ref{lem:6b}$, $TR11$.
	
	The rules used in part (\textit{v}) are $\ref{lem:6b}$, part (\textit{v}).
	
	The rules used in part (\textit{vi}) are part (\textit{i}), $\ref{lem:4}$, $TR9$, $S1$, $\ref{lem:6b}$, $\ref{lem:hopf}$, $TR1$, $K2$, part (\textit{iii}), part (\textit{v}), part (\textit{ii}).
	
	The rules used in part (\textit{vii}) are $\ref{lem:6b}$, $S1$, $TR1$, $B2$, $TR7$, $\ref{lem:4}$, $TR11$.
	
	The rules used in part (\textit{viii}) are $\ref{lem:6b}$, $H$, $B2$, $S1$, $\ref{lem:hopf}$, $\ref{lem:6}$.
	
	Finally, the proof of this proposition uses the rules $\ref{lem:4}$, $S1$, part (\textit{viii}), $\ref{lem:hopf}$, part (\textit{vii}), part (\textit{vi}), $B2$, $\ref{lem:com}$, $\ref{lem:7}$, $\ref{lem:6b}$, $TR1$.
\end{proof}
\begin{proposition}(ZW rule $nat^{m \eta}_{w}$)~\newline\\
	$
	ZX\vdash
	\left\llbracket~

	~\right\rrbracket_{WX}
	$
\end{proposition}

\begin{proof}~\newline
	Proof in \cite{Emmanuel}, proposition 7 part $5c$. The rules used are $\ref{lem:2}$, $B1$, $TR2$.
\end{proof}
\begin{proposition}(ZW rule $nat^{m \eta}_{\varepsilon ,w}$)~\newline\\
	$
	ZX\vdash
	\left\llbracket~
	%
	~\right\rrbracket_{WX}
	$
\end{proposition}

\begin{proof}~\newline
	Proof in \cite{Emmanuel}, proposition 7 part 5c. The rules used are $\ref{lem:2}$, $S1$, $B1$, $\ref{lem:inv}$.
\end{proof}
\begin{proposition}(ZW rule $hopf$)~\newline \\
	$
	ZX\vdash
	\left\llbracket~
	%
	~\right\rrbracket_{WX}
	$
\end{proposition}

\begin{proof}~\newline
	Proof in \cite{Emmanuel}, proposition 7 part $5d$. The rules used are $\ref{prop:rng-1}$, $\ref{lem:hopf}$, $S1$, $TR10$, $TR4$, $B1$.
\end{proof}
\begin{proposition}(ZW rules $sym_{3,L}$ and $sym_{3,R}$)~\newline \\
	$
	ZX\vdash
	\left\llbracket~
	%
	~\right\rrbracket_{WX}
	$
\end{proposition}

\begin{proof}
	Proof in \cite{Emmanuel}, proposition 7 part $0b$ and $0b'$ respectively. The rules used are $\ref{lem:4}$, $\ref{lem:com}$, $\ref{lem:6b}$, $S1$, $TR1$.
\end{proof}
\begin{proposition}(ZW rule $sym_{2}$)~\newline \\
	$
	ZX\vdash
	\left\llbracket~
	%
	~\right\rrbracket_{WX}
	$
\end{proposition}

\begin{proof}
	Proof in \cite{Emmanuel}, proposition 7 part $0a$. The rules used are $S1$, $\ref{lem:com}$.
\end{proof}
\begin{proposition}(ZW rule $inv$)~\newline \\
	$
	ZX\vdash
	\left\llbracket~
	%
	~\right\rrbracket_{WX}
	$
\end{proposition}

\begin{proof}
	Proof in \cite{Emmanuel}, proposition 7 part $2a$. The rule used is $S1$.
\end{proof}
\begin{proposition}(ZW rule $ant^\eta_x$)~\newline \\
	$
	ZX\vdash
	\left\llbracket~
	%
	~\right\rrbracket_{WX}
	$
\end{proposition}

\begin{proof}
	Proof in \cite{Emmanuel}, proposition 7 part $7b$. The rules used are $\ref{lem:6b}$, $H$, $S1$.
\end{proof}

\begin{proposition}(ZW rules $sym_{z,L}$ and $sym_{z,R}$)~\newline \\
	$
	ZX\vdash
	\left\llbracket~
	%
	~\right\rrbracket_{WX}
	$
\end{proposition}

\begin{proof}
	Proof follows directly from rules $S1$ and $\ref{lem:com}$.
\end{proof}
\begin{proposition}(ZW rule $un^{co}_{z,R}$)~\newline \\
	$
	ZX\vdash
	\left\llbracket~
	%
	~\right\rrbracket_{WX}
	$
\end{proposition}

\begin{proof}
	Proof follows directly from rule $S1$.
\end{proof}
\begin{proposition}(ZW rule $nat^z_z$)~\newline \\
	$
	ZX\vdash
	\left\llbracket~
	%
	~\right\rrbracket_{WX}
	$
\end{proposition}

\begin{proof}
	Proof follows directly from rules $S1$ and $\ref{lem:com}$.
\end{proof}
\begin{proposition}\label{prop:ph}(ZW rule $ph$)~\newline \\
	$
	ZX\vdash
	\left\llbracket~
	%
	~\right\rrbracket_{WX}
	$
\end{proposition}

\begin{proof}
	Proof in \cite{Emmanuel}, proposition 7 part $3a$. The rule used is $\ref{lem:6b}$.
\end{proof}
\begin{proposition}(ZW rule $nat^m_c$)~\newline \\
	$
	ZX\vdash
	\left\llbracket~
	%
	~\right\rrbracket_{WX}
	$
\end{proposition}

\begin{proof}
	Proof in \cite{Emmanuel}, proposition 7 part $6a$. The rules used are $S1$, $B2$, $TR10$.
\end{proof}
\begin{proposition}(ZW rule $loop$)~\newline \\
	$
	ZX\vdash
	\left\llbracket~
	%
	~\right\rrbracket_{WX}
	$
\end{proposition}

\begin{proof}
	Proof in \cite{Emmanuel}, proposition 7 part $6c$. The rules used are $\ref{lem:hopf}$, $B1$, $TR2$.
\end{proof}
\begin{proposition}(ZW rule $unx$)~\newline \\
	$
	ZX\vdash
	\left\llbracket~
	%
	~\right\rrbracket_{WX}
	$
\end{proposition}

\begin{proof}
	Proof in \cite{Emmanuel}, proposition 7 part $X$. The rules used are $\ref{lem:1}$, $S1$, $B2$, $\ref{lem:6}$, $TR5$.
\end{proof}
\begin{proposition}(ZW rule $rng_1$)~\newline \\
	$
	ZX\vdash
	\left\llbracket~
	%
	~\right\rrbracket_{WX}
	$
\end{proposition}

\begin{proof}
	Proof follows directly from rule $L3$.
\end{proof}
\begin{proposition}\label{prop:rng-1}(ZW rule $rng_{-1}$)~\newline \\
	$
	ZX\vdash
	\left\llbracket~
	%
	~\right\rrbracket_{WX}
	$
\end{proposition}

\begin{proof}
	Proof follows directly from rules $S1$, $L4$, $L5$.
\end{proof}
\begin{proposition}(ZW rule $rng^{r,s}_{+}$)~\newline \\
	$
	ZX\vdash
	\left\llbracket~
	%
	~\right\rrbracket_{WX}
	$
\end{proposition}

\begin{proof}
	Proof follows directly from rule $AD$.
\end{proof}
\begin{proposition}(ZW rule $nat^{r}_{c}$)~\newline \\
	$
	ZX\vdash
	\left\llbracket~
	%
	~\right\rrbracket_{WX}
	$
\end{proposition}

\begin{proof}
	Proof follows directly from rules $TR13$, $TR14$.
\end{proof}
\begin{proposition}(ZW rule $nat^{r}_{\varepsilon c}$)~\newline \\
	$
	ZX\vdash
	\left\llbracket~
	%
	~\right\rrbracket_{WX}
	$
\end{proposition}

\begin{proof}
	Proof follows directly from rules $S4$, $L2$.
\end{proof}
\begin{proposition}(ZW rule $ph^{r}$)~\newline \\
	$
	ZX\vdash
	\left\llbracket~
	%
	~\right\rrbracket_{WX}
	$
\end{proposition}

\begin{proof}
	Proof follows directly from rules $S1$, $L1$.
\end{proof}


\begin{thebibliography}{99}
\bibitem{CoeckeDuncan} B. Coecke,  R. Duncan (2011): Interacting quantum
observables: Categorical algebra and diagrammatics. New Journal of Physics 13, p.
043016.

\bibitem{Joyal} Andr$\acute{e}$ Joyal and Ross Street. The geometry of tensor calculus, I. Advances in Mathematics, 88(1):55-112, July 1991.


\bibitem{Coeckesamson} Samson Abramsky and Bob Coecke, A categorical semantics of quantum protocols, Proceedings of the 19th IEEE conference on Logic in Computer Science (LiCS'04), IEEE Computer Science Press (2004).


\bibitem{Duncanpx} Duncan, R., and Perdrix, S. 2010. Rewriting measurement-based quantum computations with generalised flow. Pages 285-296 of: Proceedings of ICALP. Lecture Notes in Computer Science. Springer.

\bibitem{Horsman} Horsman, C. 2011. Quantum picturalism for topological cluster-state computing. New Journal of Physics, 13, 095011. 

\bibitem{DuncanLucas} Duncan, R., and Lucas, M. 2013. Verifying the Steane code with Quantomatic. In: Proceedings of the 10th International Workshop on Quantum Physics and Logic. arXiv:1306.4532.


\bibitem{ckzh}Nicholas Chancellor, Aleks Kissinger, Stefan Zohren, Dominic Horsman, Coherent Parity Check Construction for Quantum Error Correction, arXiv:1611.08012

\bibitem{Zamdzhiev} C. Schr\"oder de Witt, V. Zamdzhiev. The ZX-calculus is incomplete for quantum mechanics. EPTCS 172, pp.285-292, 2014.

\bibitem{duncanperdrix} Ross Duncan and Simon Perdrix,  Pivoting makes the ZX-calculus complete for real stabilizers,  Electronic Proceedings in Theoretical Computer Science, 2013.


\bibitem{Miriam1} Miriam Backens (2014): The ZX-calculus is complete for stabilizer quantum mechanics. In: New Journal of Physics. Vol. 16. No. 9. Pages 093021.


\bibitem{Miriam1ct}Miriam Backens, The ZX-calculus is complete for the single-qubit Clifford+T group, Electronic Proceedings in Theoretical Computer Science, 2014.









\bibitem{Emmanuel} Emmanuel Jeandel, Simon Perdrix, Renaud Vilmart, A Complete Axiomatisation of the ZX-Calculus for Clifford+T Quantum Mechanics, arXiv:1705.11151


\bibitem{amar} Amar Hadzihasanovic, The algebra of entanglement and the geometry of composition, PhD thesis.

\bibitem{bpw} Miriam Backens, Simon Perdrix, and Quanlong Wang, A simplified stabilizer ZX-calculus, Electronic Proceedings in Theoretical Computer Science, 2016.



























\bibitem{Miriam2} Miriam Backens, Completeness and the ZX-calculus, PhD thesis,  arXiv:1602.08954.
























\bibitem{Quanto} Quantomatic. https://sites.google.com/site/quantomatic/

\end{thebibliography}
\end{document}